\documentclass[preprint,11pt,authoryear]{elsarticle}
\usepackage{tikz}
\usepackage{xcolor}
\usepackage{amsfonts}
\usepackage{amsmath}
\usepackage{amsthm}
\usepackage{amssymb}
\usepackage{times}
\usepackage{subfig}
\usepackage{multirow}
\usepackage{setspace}
\usepackage{enumerate}
\usepackage{natbib}
\usepackage{color}
\usepackage{colortbl}
\usepackage{verbatim}
 \usepackage{lscape}
\usepackage{array}
\usepackage{setspace}
\usepackage[normalem]{ulem}

\usepackage{caption}
\usepackage{rotating}

\newcolumntype{x}[1]{%
>{\centering\hspace{0pt}}p{#1}}%

\def\iid{\buildrel {\rm i.i.d.} \over \sim}

\def\i.i.d.{\buildrel {\rm i.i.d.} \over \sim}

\def\cw#1 { \overset{\mathbb{P}}{\underset{#1}{\longrightarrow}} }
\def\Real{\mathbb{R}}
\def\P#1{{\mathbb{P}}\left(#1\right)}

\def\E#1{{\mathbb E}\left[#1\right]}

\def \rcov#1#2 {{\rm cov}_{#1}\left( #2\right)}

\newtheorem{example}{Example}

\newtheorem{lemma}{Lemma}
\newtheorem{theorem}{Theorem}
\newtheorem{definition}{Definition}
\newtheorem{corollary}{Corollary}
\newtheorem{remark}{Remark}
\newtheorem{proposition}{Proposition}

\newtheorem*{toy*}{Finite Model}

\newtheorem{model}{Model}
\newtheorem{cond.model}{Conditional Model}

\def\cov#1{{\rm  cov}\left[#1\right]}

\begin{document}
\begin{frontmatter}
\title{On Copula-based Collective Risk Models}

\author[EH]{Rosy Oh}
\ead{rosy.oh5@gmail.com}
\address[EH]{Department of Statistics, Ewha Womans University, 11-1 Daehyun-Dong, Seodaemun-Gu, Seoul 120-750, Korea.}
\author[EH]{Jae Youn Ahn\corref{cor2}}
\ead{jaeyahn@ewha.ac.kr}
\author[IH]{Woojoo Lee\corref{cor2}}
\ead{lwj221@gmail.com}
\address[IH]{Department of Statistics, Inha University, 235 Yonghyun-Dong, Nam-Gu, Incheon 402-751, Korea.}

\cortext[cor2]{Corresponding Authors}

\begin{abstract}

Several collective risk models have recently been proposed by relaxing the widely used but controversial assumption of independence between claim frequency and severity.
Approaches include the bivariate copula model, random effect model, and two-part frequency-severity model.
This study focuses on the copula approach to develop collective risk models that allow a flexible dependence structure for frequency and severity.
We first revisit the bivariate copula method for frequency and average severity. After examining the inherent difficulties of the bivariate copula model, we alternatively propose modeling the dependence of frequency and individual severities using multivariate Gaussian and t-copula functions. The proposed copula models have computational advantages and provide intuitive interpretations for the dependence structure. Our analytical findings are illustrated by analyzing automobile insurance data.
\end{abstract}


\end{frontmatter}

\vfill

\pagebreak

\vfill

\pagebreak

\section{Introduction}

The collective risk model, defined as the sum of the severities or the average of the severities,
is an important tool for decision making in the insurance sector.
Traditionally, two types of independence assumptions are assumed in collective risk models: one is independence between claim frequency and each individual severity and the other is independence among individual severities, as discussed by \citet{Klugman}. In this paper, we call a collective risk model with these two independence assumptions the ``independent collective risk model.''

Recently, researchers have relaxed those independence assumptions using flexible statistical models such as a shared random effect model \citep{Bastida, Czado2015} and a copula model \citep{Czado, Kramer2013, Gee2016, Marceau2018}. Specifically, the copula models of \citet{Czado} and \citet{Kramer2013} introduce the dependence between frequency $(N)$ and average severity $(M)$ via a parametric copula family including a Gaussian copula; these models are also used by \citet{Gee2016}, \citet{Marceau2018}, and \citet{Leegee}.
Throughout this paper, we call $(N,M)$ summarized data.
Alternatively, \citet{Frees2} propose the so-called {\it two-step frequency-severity model} to provide the dependence between frequency and severity. Various applications of their two-step frequency-severity model can be found in \citet{Peng}, \citet{Garrido}, \citet{park2018}, and \citet{AhnValdez2}.

Throughout this paper, we call $(N, Y_1, \cdots, Y_N)$ micro-level data, where $Y_{i}$ denotes the individual severity.
Dependence models for micro-level data have also been studied in the literature.
For example, under a Sparre Andersen-type dependence structure, the dependence between frequency and severity is explained
using the dependence between interclaim time and severity \citep{Albrecher2006, Boudreault2006, Cossette2008, Cossette2010, Asimit2010, Landriault2014}.
Recently, \citet{Liu2017} and \citet{Marceau2018}
present copula models for micro-level data focusing on the structural property of the aggregate sum.
However, important issues about parameter estimation and how to interpret the dependence structure have not been investigated in depth.
In addition, the existing literature does not discuss
the important differences between micro-level and summarized data when modeling the dependence between frequency and severity.

In this paper,
we first explain the need for a copula model for micro-level data by explaining the inherent difficulties in finding a suitable copula function for summarized data.
Specifically, we show that even under the independent collective risk model, existing copula families may not suitably describe the empirical properties of summarized data because of the intrinsic dependence structure between $N$ and $M$. As an alternative to the copula model for summarized data,
we introduce the Gaussian and t-copula models for analyzing micro-level data.
To make these models concrete, we investigate useful correlation matrices to define the dependence in the copula model and find conditions for the correlation matrix to be positive definite. Using the proposed correlation matrix and corresponding Gaussian or t-copula, we show how to model the various dependence structures in the collective risk model. Our model can accommodate two types of dependence: the dependence between claim frequency and each individual severity and the dependence among individual severities. In addition, we show how to extend the proposed copula model to accommodate the regression setting.
The proposed models have computational advantages in terms of parameter estimation and provide an intuitive interpretation of the dependence structure.

The remainder of this paper is organized as follows. Section 2 defines the frequently used notations. The difficulties in finding a suitable copula function for summarized data are
explained in Section 3. Before proposing our copula model, we first study the correlation matrices used to explain the dependence structure between
frequency and individual severities and that among individual severities. In particular, conditions are provided to guarantee when they are positive definite in Section 4.
Section 5 deals with a copula model for positive frequency data and Section 6 extends it to observed data including zero frequency. Some regression settings are also discussed.
The numerical study is described in Section 7 and our analytical findings are illustrated by analyzing automobile insurance data in Section 8, followed by concluding remarks.

\section{Symbols}\label{Model}
Let $\mathcal{N}$ be a set of positive integers and $\mathcal{N}_0$ be a set of non-negative integers.
Let $N_{i}$ represent the number of claims ({\it frequency}) of the $i$-th policyholder and
$Y_{ij}$ indicate the claim size ({\it individual severity}) in the $j$-th claim of the $i$-th policyholder.
For a non-negative integer $k$, we define
$$\boldsymbol{Y}_{i}^{[k]}:=\begin{cases}
\left( Y_{i1}, \cdots, Y_{ik} \right)^{\mathrm T}, & k>0;\\
\hbox{null}, & k=0.\\
\end{cases}$$
We further define two quantities:
 \begin{equation}\label{Eqq.1}
   S_i:=
   \begin{cases}
        \sum\limits_{j=1}^{N_i} Y_{ij}, & N_i>0;\\
   \hbox{null}, & N_i=0;\\
   \end{cases}
   \quad\hbox{and}\quad
   M_i:=
   \begin{cases}
        \frac{\sum\limits_{j=1}^{N_i} Y_{ij}}{N_i}, & N_i>0;\\
   \hbox{not defined}, & N_i=0.\\
   \end{cases}
 \end{equation}
 Here, $S_i$ and $M_i$ are called the {\it aggregated severity} and {\it average severity}, respectively.
They are linked as
$$M_i=\frac{S_i}{N_i}\quad , N_i>0.$$
We use $n_i$, $\boldsymbol{y}_i^{[k]}$, $s_i$, and $m_i$ as the realization of $N_i$, $\boldsymbol{Y}_i^{[k]}$, $S_i$, and $M_i$, respectively.

For the given frequencies $\left\{n_1, \cdots, n_l\right\}$ from $l$ policyholders,
define
  \begin{equation*}
  \mathcal{I}_l:=\left\{i \in \left\{1, \cdots, l \right\} \big\vert n_i\neq 0 \right\}.
  \end{equation*}
Furthermore, we call
\begin{equation}\label{eq.1}
\left\{(n_i, \boldsymbol{y}_i^{\mathrm{T}})\big\vert i=1, \cdots, l\right\}
\end{equation}
full data and
\begin{equation}\label{eq.2}
\left\{(n_i, m_i)\big\vert i=1, \cdots, l\right\}
\end{equation}
summarized data. Summarized data \eqref{eq.2} are understood as
\begin{equation}\label{eq.3}
\left\{n_i\big\vert i\notin \mathcal{I}_l\right\} \cup \left\{(n_i, m_i)\big\vert i\in\mathcal{I}_l\right\}
\end{equation}
because $m_i$ is not defined when $n_i=0$. When the context is clear, we drop the subscript $i$ to simplify the notations.
For example, we denote $m_i$, $n_i$, and $\boldsymbol{y}_i^{[n_i]}$ by $m$, $n$, and $\boldsymbol{y}^{[n]}$, respectively
if it is clear that these notations are defined for the $i$-th policyholder.

For the frequency part, we allow any non-negative integer-valued distribution including distributions in the (reproductive) {\it exponential dispersion family} (EDF) and zero-inflated count distributions \citep{yip2005modeling}. We use $F_1(x; \lambda, \psi_1)$ and $f_1(x; \lambda, \psi_1)$ to denote the cumulative distribution function and probability mass function, respectively. Here, $\lambda$ and $\psi_1$ correspond to the parameter of interest and nuisance parameter(s), respectively.
For the severity part, to simplify the model, we only consider continuous positive distributions with a probability density function, including
distributions belonging to the continuous EDF and heavy-tailed distributions.
 $F_2(x; \xi, \psi_2)$ and $f_2(x; \xi, \psi_2)$ are used to denote the cumulative distribution function and probability density function, respectively. Similar to the frequency part, $\xi$ and $\psi_2$ are the parameter of interest and nuisance parameter(s), respectively.
In a clear context, we simply use $F_1$, $f_1$, $F_2$, and $f_2$ for $F_1(x; \lambda, \psi_1)$, $f_1(x; \lambda, \psi_1)$, $F_2(x; \xi, \psi_2)$, and $f_2(x; \xi, \psi_2)$, respectively.

\section{Dependence in Collective Risk Models}

One of the key assumptions frequently used in classical collective risk models is the independence of frequency and individual severities and the independence assumption among individual severities.
However, recent studies \citep{Czado,Kramer2013, Frees2, Czado2015, Peng, Garrido, AhnGLM, park2018, AhnValdez2}
have reported evidence against the independence assumption.

To capture the dependence between frequency and severity or among individual severities, \citet{Bastida} and \citet{Czado2015} use a shared random effect model, and \citet{Frees2}, \citet{Peng}, \citet{Garrido}, \citet{AhnGLM}, \citet{park2018}, and \citet{AhnValdez2} use a frequency model to predict severities in the regression setting.
 On the contrary, \citet{Czado}, \citet{Kramer2013}, \citet{Gee2016}, \citet{Marceau2018}, and \citet{Leegee} adopt a parametric copula approach, including a Gaussian copula, to show the dependence between frequency and average severity. While the copula is a widely used tool for modeling dependence, the choice of a suitable copula family is often a more difficult problem than the choice of a suitable marginal distribution family. In particular, when modeling the dependence between frequency and average severity, the choice of a suitable copula family can be even harder.
 The following example shows that most existing copula families including Gaussian and Archimedean copulas cannot accommodate the dependence between frequency and average severity properly, even under the simplest assumption where frequency and individual severities are assumed to be independent.
\begin{example}\label{ex.1}
Consider the classical collective risk model, where frequency $N$ and the individual severity $Y_i$'s are assumed to be independent. We further assume that
$N$ is a zero-truncated Poisson distribution with
\[
\P{N=n}=\frac{\lambda^n}{(e^\lambda-1)n!}
\]
and
\[
Y_1, \cdots, Y_N \big\vert N \iid {\rm Gamma}(\xi, \psi).
\]
Then,
we have
\[
M\big\vert N\sim {\rm Gamma}(\xi, \psi/N ).
\]
Clearly, $N$ and $M$ are not independent even though frequency and individual severities are independent.

Now, we want to visualize the density function of a suitable copula family for $(N, M)$
under the assumption that frequency and individual severities are independent.
Let $F_{N}$ and $F_{M}$ denote the distribution functions for $N$ and $M$, respectively. $\rm Ran F$ means range of $F$.
Since the copula of $(N, M)$ is unique only on ${\rm Ran}F_{N}\times {\rm Ran}F_{M}$ as shown by \citet{Sklar},
the corresponding copula density function is not easily visualized. Instead, we define the alternative random vector $(N^*, M)$ as
\begin{equation}\label{eq..1}
N^*:=N+Z \quad\hbox{and}\quad M\big\vert N\sim {\rm Gamma}(\xi, \psi/N ),
\end{equation}
where $Z\sim {\rm Unif}[0,1]$ and $Z$ are independent of $N$ and $(Y_1, \cdots, Y_N)$. Clearly, $(N^*, M)$ is a continuous random vector, and the corresponding copula is uniquely determined on $[0,1]\times[0,1]$. While the copula of $(N^*, M)$ is different to the corresponding (sub)copula of $(N, M)$, we can have an important insight
into the shape of the (sub)copula of $(N, M)$ by examining that of $(N^*, M)$.

Since the corresponding copula of $(N^*, M)$ is implicitly defined, in this example, we estimate the corresponding copula using a kernel density estimation of the copula with simulated samples \citep{Gijbels,Chen}.
Figure \ref{figu.1} shows the kernel density function of the copula $\widehat{C}^*$, using $n=2,000$ pairs of the i.i.d. random vector from $(N^*,M)$, and the corresponding contour plot.

  \begin{figure}[h]
     \centering
    \subfloat[3D Plot: Density of the Copula]{%
      \includegraphics[width=0.49\textwidth]{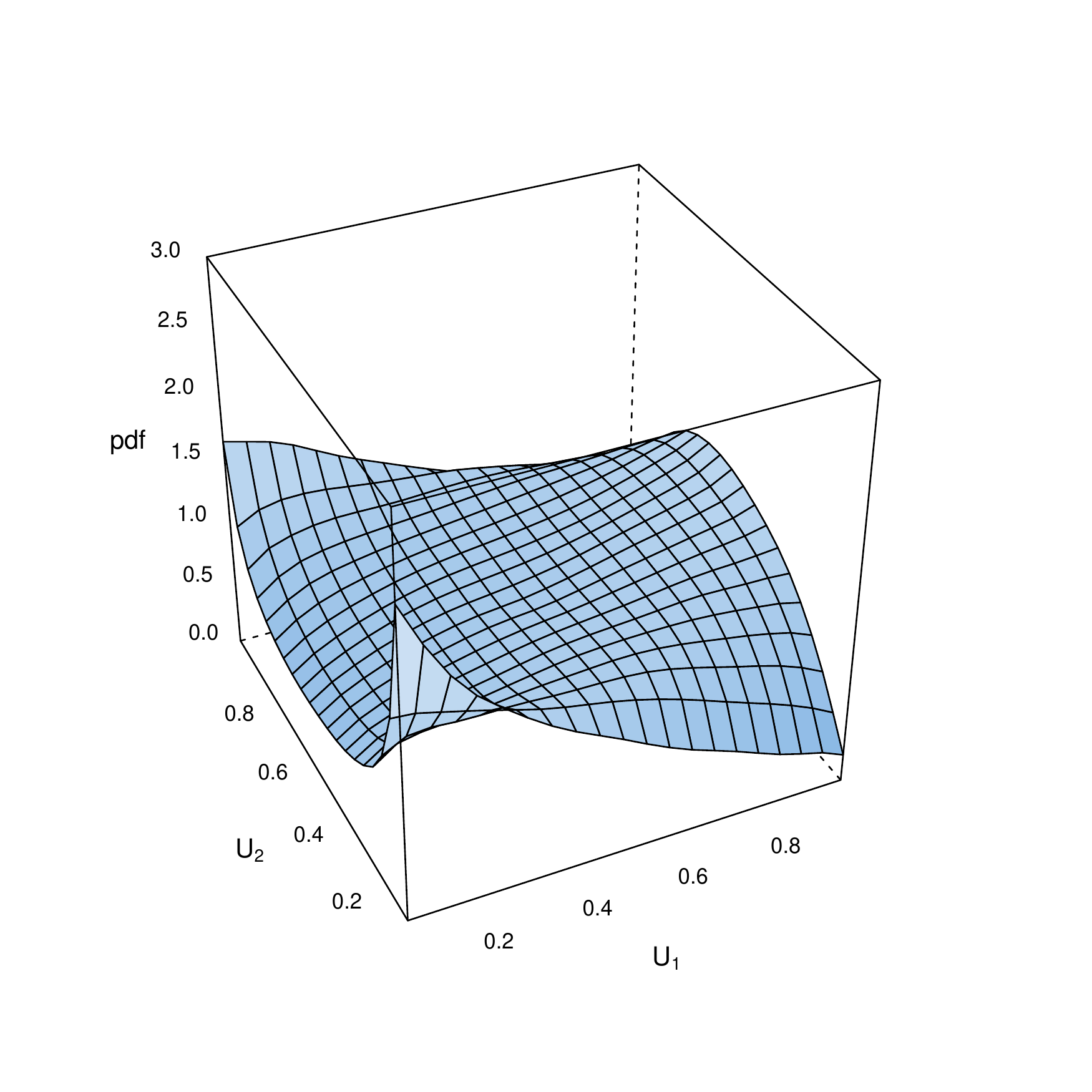}
    }
    \hfill
    \subfloat[Contour: Density of the Copula]{%
      \includegraphics[width=0.49\textwidth]{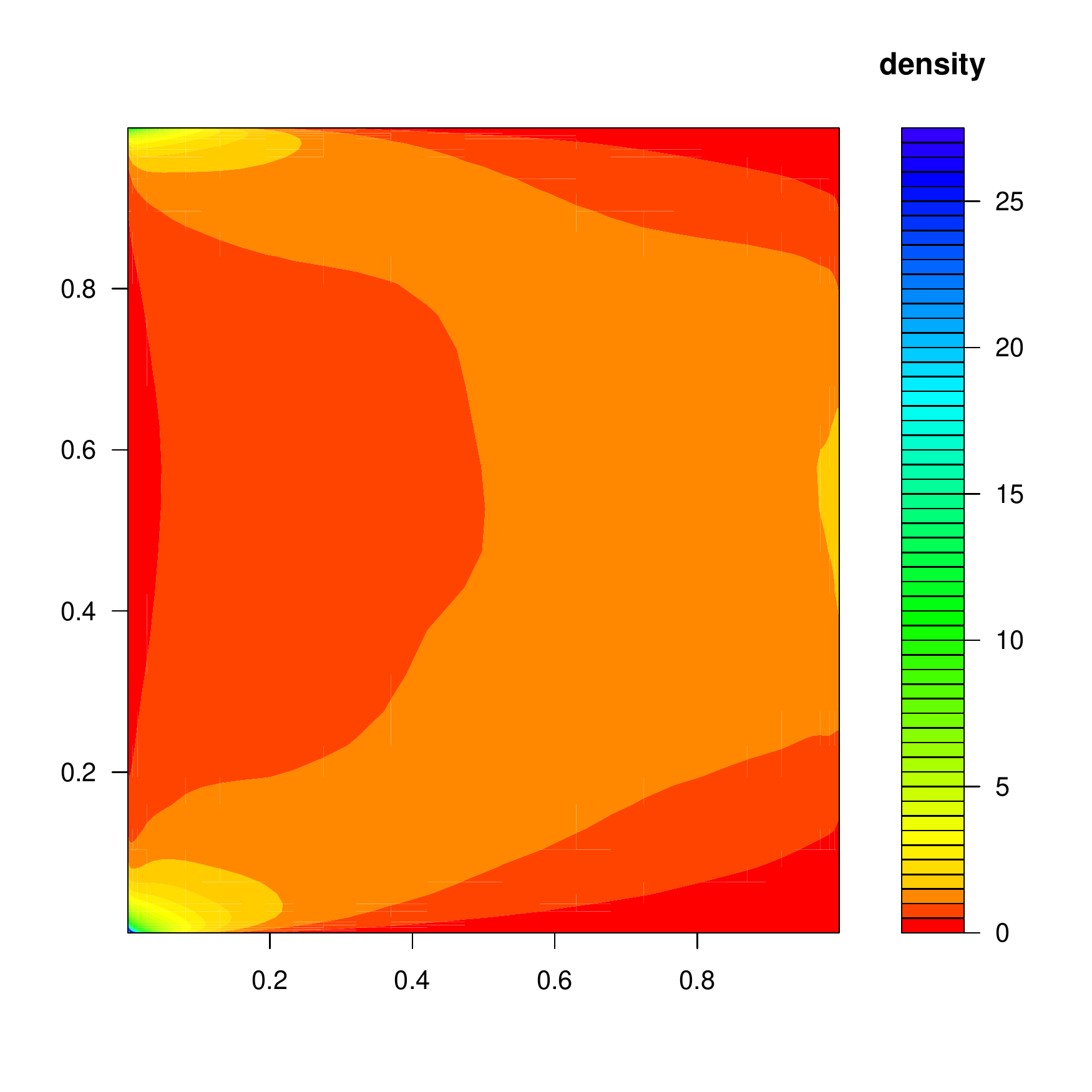}
    }
      \caption{Density estimation of the copula of $(N^*,M)$ using a kernel density estimation}
     \label{figu.1}
  \end{figure}
Let $(U_1, U_2)$ be a random vector sampled from $\widehat{C}^*$ in Figure \ref{figu.1}.
 As shown from the figure, for the lower $U_1$, the density of the copula tends to be smaller at the center of $U_2$. On the contrary, for the higher $U_1$, the density of the copula tends to be larger at the center of $U_2$. Clearly, the density plots in Figure \ref{figu.1} reflect the fact that the conditional variance of $M$ shrinks as $N$ rises, which captures the most eminent feature of the copula of $(N, M)$.
\end{example}

In conclusion, the choice of the copula family to provide a suitable dependence structure between $N$ and $M$ described in Example \ref{ex.1} can be difficult.
We consider that no existing parametric copula family can reflect the property described in Example \ref{ex.1} accurately.
One may consider using the non-parametric bivariate copula approach as in \citet{Chen2007}; however, this also has difficulties to
provide a straightforward interpretation of the dependence, as shown in Example \ref{ex.1}, as long as $(N,M)$ is modeled.

In the subsequent sections, rather than directly providing the dependence of the summarized data $(N,M)$, we provide the dependence structure of the micro-level data $(N, Y_1, \cdots, Y_N)$ using a copula method. Such an approach requires access to full data, whereas the approach in Example \ref{ex.1} only requires the summarized data.

\section{Correlation Matrix for Frequency and Individual Severities}\label{sec.4}

This section presents two useful correlation matrices to describe the dependence of $(N, Y_1, \cdots, Y_N)$.

\subsection{Equicorrelation matrix}

We study a correlation matrix that has a common pairwise correlation for individual severities $(Y_1, \cdots, Y_N)$.
Based on this, we investigate an extended correlation matrix for $(N, Y_1, \cdots, Y_N)$.

\begin{definition}\label{def.1}
  For $\rho_1, \rho_2\in[-1,1]$, define the following matrices.
  \begin{enumerate}
    \item[i.] For any positive integer $k$, define a $k\times k$ matrix $\boldsymbol{\Sigma}_{\rho_2}^{[k, 1]}$
  \begin{equation*}
  \left[\boldsymbol{\Sigma}_{\rho_2}^{[k,1]}\right]_{i,j}:=\begin{cases}
    1, &\hbox{if}\quad i=j\\
    \rho_2, &\hbox{if}\quad i\neq j\\
  \end{cases}
  \end{equation*}
  for $i,j=1, \cdots, k$.
    \item[ii.] For any non-negative integer $k$, define a $(k+1)\times(k+1)$ matrix $\boldsymbol{\Sigma}_{\rho_1,\rho_2}^{[k,1]}$
  \begin{equation*}
    \boldsymbol{\Sigma}_{\rho_1,\rho_2}^{[k, 1]}:=\begin{cases}
    \left(
    \begin{array}{cc}
      1 & \rho_1 \left({\bf 1}_k\right)^{\mathrm T}\\
      \rho_1 {\bf 1}_k & \boldsymbol{\Sigma}_{\rho_2}^{[k,1]}\\
    \end{array}
    \right), & k=1,2, \cdots, ;\\
     1, & k=0;\\
    \end{cases}
  \end{equation*}
    where ${\bf 1}_k$ is a column vector of $1$ with length $k$.

  \end{enumerate}

\end{definition}
In matrix form, $\boldsymbol{\Sigma}_{\rho_2}^{[k,1]}$ and $\boldsymbol{\Sigma}_{\rho_1,\rho_2}^{[k, 1]}$ in Definition \ref{def.1} are written as
  \[
  \boldsymbol{\Sigma}_{\rho_2}^{[k,1]}=\left(
                      \begin{array}{ccccc}
                        1 & \rho_2 & \rho_2 & \cdots & \rho_2 \\
                        \rho_2 & 1 & \rho_2 & \cdots & \rho_2 \\
                        \rho_2 & \rho_2 & 1 & \cdots & \rho_2 \\
                        \vdots & \vdots & \vdots & \ddots & \vdots \\
                        \rho_2 & \rho_2 & \rho_2 & \cdots & 1 \\
                      \end{array}
                    \right)\quad\hbox{and}\quad
                      \boldsymbol{\Sigma}_{\rho_1,\rho_2}^{[k,1]}=\left(
                      \begin{array}{cccccc}
                        1 & \rho_1 & \rho_1 & \rho_1 & \cdots & \rho_1 \\
                        \rho_1 & 1 & \rho_2 & \rho_2 &\cdots & \rho_2 \\
                        \rho_1 & \rho_2 & 1 & \rho_2 &\cdots & \rho_2 \\
                        \rho_1 & \rho_2 & \rho_2 & 1 &\cdots & \rho_2 \\
                        \vdots & \vdots & \vdots & \vdots & \ddots & \vdots \\
                        \rho_1 & \rho_2 & \rho_2 & \rho_2 & \cdots & 1 \\
                      \end{array}
                    \right).
  \]

The following proposition provides the determinants of $\boldsymbol{\Sigma}_{\rho_2}^{[k,1]}$ and $\boldsymbol{\Sigma}_{\rho_1,\rho_2}^{[k,1]}$.

\begin{proposition}\label{prop.oh.1}
For any non-negative integer $k$
and $\rho_1, \rho_2\in(-1,1)$, we have
  \begin{equation}\label{eq.31}
  {\rm det}\left(\boldsymbol{\Sigma}_{\rho_2}^{[k,1]} \right)=
    \left( 1+ (k-1)\rho_2\right) \left( 1-\rho_2\right)^{k-1}
  \end{equation}
  and
  \begin{equation*}
  {\rm det}\left(\boldsymbol{\Sigma}_{\rho_1,\rho_2}^{[k,1]} \right)=
    \left[1+(k-1)\rho_2-k(\rho_1)^2 \right](1-\rho_2)^{k-1}.
  \end{equation*}
\end{proposition}

The proof is given in Appendix A.
We also provide the inverse of $\boldsymbol{\Sigma}_{\rho_2}^{[k,1]} $ and $\boldsymbol{\Sigma}_{\rho_1,\rho_2}^{[k,1]} $.
\begin{proposition}\label{prop.oh.5}
Let $k$ be any non-negative integer
and $\rho_1, \rho_2\in(-1,1)$.
If ${\rm det}\left(\boldsymbol{\Sigma}_{\rho_2}^{[k,1]} \right)\neq 0$, then
 the inverse matrix of $\boldsymbol{\Sigma}_{\rho_2}^{[k,1]}$ is
    \begin{equation}\label{eq.32}
  \left( \boldsymbol{\Sigma}_{\rho_2}^{[k,1]}\right)^{-1}=\frac{1}{1-\rho_2} \left[\boldsymbol{I}_k -\frac{\rho_2}{1+(k-1)\rho_2} \boldsymbol{J}_{k\times k} \right],
  \end{equation}
  where $\boldsymbol{J}_{k\times k}$ is a $k \times k$ matrix of ones.
Furthermore, if ${\rm det}\left(\boldsymbol{\Sigma}_{\rho_1,\rho_2}^{[k,1]} \right)\neq 0$, then
  \begin{equation}\label{eq.oh.110}
  \left[ \boldsymbol{\Sigma}_{\rho_1,\rho_2}^{[k,1]} \right] ^{-1}=
  \left(
    \begin{array}{cc}
      \frac{\rho_1^2 k }{ 1+(k-1)\rho_2 - k\rho_1^2 } & \frac{\rho_1 }{ 1+(k-1)\rho_2 - k\rho_1^2 }\boldsymbol{1}_k^{\mathrm T} \\
      \frac{\rho_1 }{ 1+(k-1)\rho_2 - k\rho_1^2 }\boldsymbol{1}_k & \frac{1}{1-\rho_2}\left[\boldsymbol{I}_k - \frac{\rho_2-\rho_1^2}{1+(k-1)\rho_2-k\rho_1^2} \boldsymbol{J}_{k\times k} \right] \\
    \end{array}
  \right).
  \end{equation}
\end{proposition}

The proof is given in Appendix A.
The following theorem provides the condition for $\boldsymbol{\Sigma}_{\rho_2}^{[k,1]}$ and $\boldsymbol{\Sigma}_{\rho_1,\rho_2}^{[k,1]}$ to be positive definite.
\begin{theorem}\label{thm.1}
Let $k$ be any positive integer and $\rho_1, \rho_2\in(-1,1)$.
Then,
  $\boldsymbol{\Sigma}_{\rho_2}^{[k,1]}$ is positive definite if and only if $\rho_2$ satisfies
    \begin{equation}\label{eq.52}
    1+\rho_2(k-1)>0.
    \end{equation}
Similarly,
  $\boldsymbol{\Sigma}_{\rho_1,\rho_2}^{[k,1]}$ is positive definite if and only if $\rho_1$ and $\rho_2$ satisfy
    \begin{equation}\label{eq.10}
       k(\rho_1)^2-1  < (k-1)\rho_2.
    \end{equation}
\end{theorem}

The proof is given in Appendix A.
For any positive integer $k$ and $\rho_1\in(-1,1)$, define
    \[
    L_{\rho_1,1}(k)=\begin{cases}
      \frac{k(\rho_1)^2-1}{k-1 }, & k>1;\\
      -1, & \hbox{otherwise}.
    \end{cases}
    \]
In Figure \ref{fig1}, the shaded area of $(\rho_1, \rho_2)$ guarantees that $\boldsymbol{\Sigma}_{\rho_1,\rho_2}^{[k,1]}$ is positive definite for different $k$. As shown in the figure, the shaded area shrinks as $k$ increases.
The following corollary, which extends Theorem \ref{thm.1}, formally describes such an observation.

\begin{corollary}\label{cor.1}
Let $k_1, \cdots, k_z$ be non-negative integers.
Then, we have the following results.
\begin{enumerate}
  \item[i.] For $\rho_1, \rho_2\in[-1,1]$, the correlation matrices
  $$\boldsymbol{\Sigma}_{\rho_1, \rho_2}^{[k_1,1]}, \cdots, \boldsymbol{\Sigma}_{\rho_1, \rho_2}^{[k_z,1]}$$ are positive definite if and only if $\rho_1$ and $\rho_2$ satisfy
    \begin{equation}
   L_{\rho_1,1}(\max\{ k_1, \cdots, k_z\})< \rho_2 < 1.
    \end{equation}
  \item[ii.] If $\rho_1$ and $\rho_2$ satisfy
  \[
  \rho_1^2<\rho_2<1,
  \]
  then
  \[
  \boldsymbol{\Sigma}_{\rho_1, \rho_2}^{[k,1]}
  \]
  is positive definite for any non-negative integer $k$.
\end{enumerate}

\end{corollary}

For the proof, the first part comes from the fact that for $\rho_1\in(-1, 1)$, $L_{\rho_1,1}(k)$
  is a non-decreasing function of $k$ for $k\ge1$. The second part is the result from
  \[
\lim\limits_{k\rightarrow \infty}L_{\rho_1,1}(k)=\rho_1^2
\]
  and the first part.

\begin{figure}[!ht]
  \centering
  \subfloat[$k=3$]{\includegraphics[width=.33\textwidth]{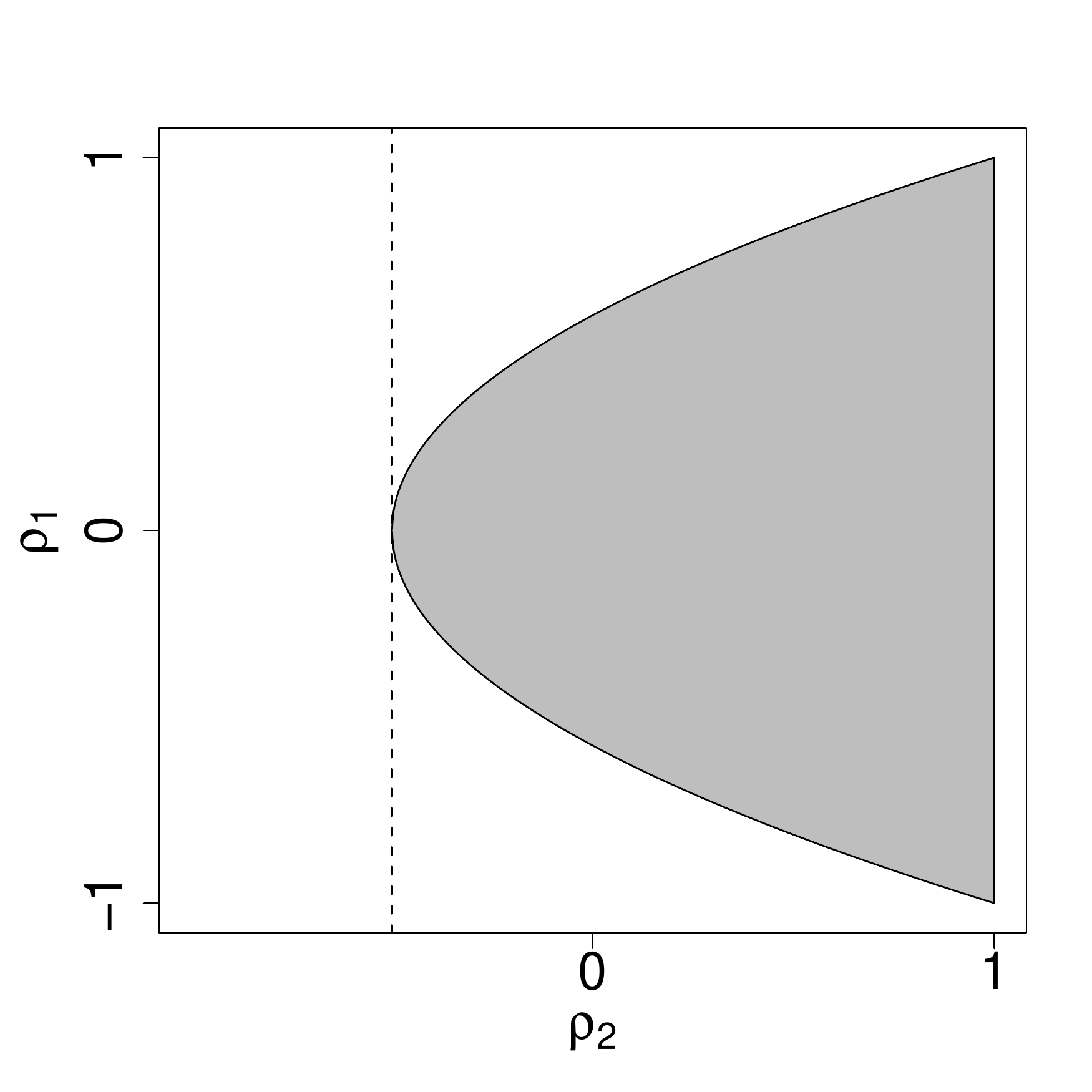}}
  \subfloat[$k=5$]{\includegraphics[width=.33\textwidth]{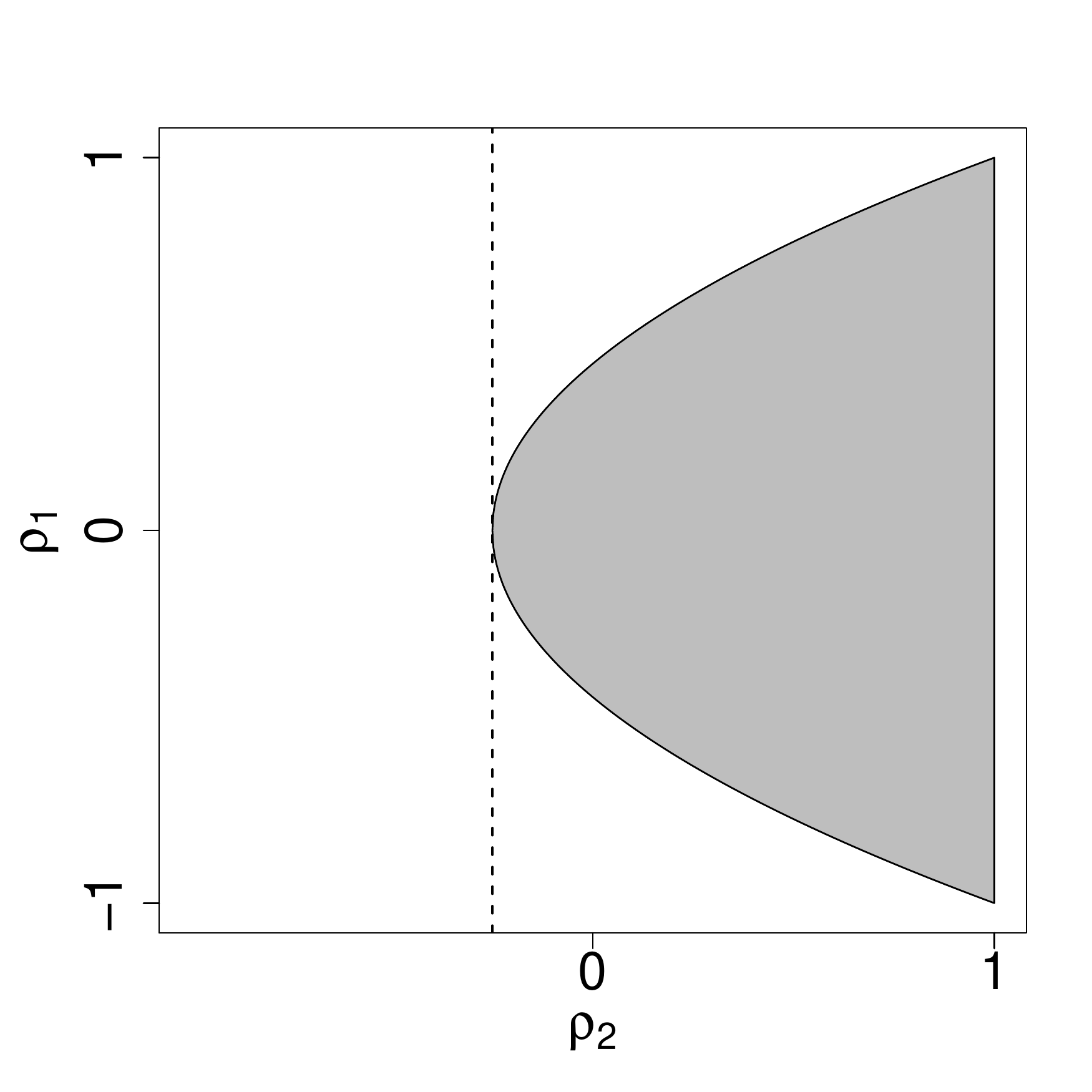}}
  \subfloat[$k=10$]{\includegraphics[width=.33\textwidth]{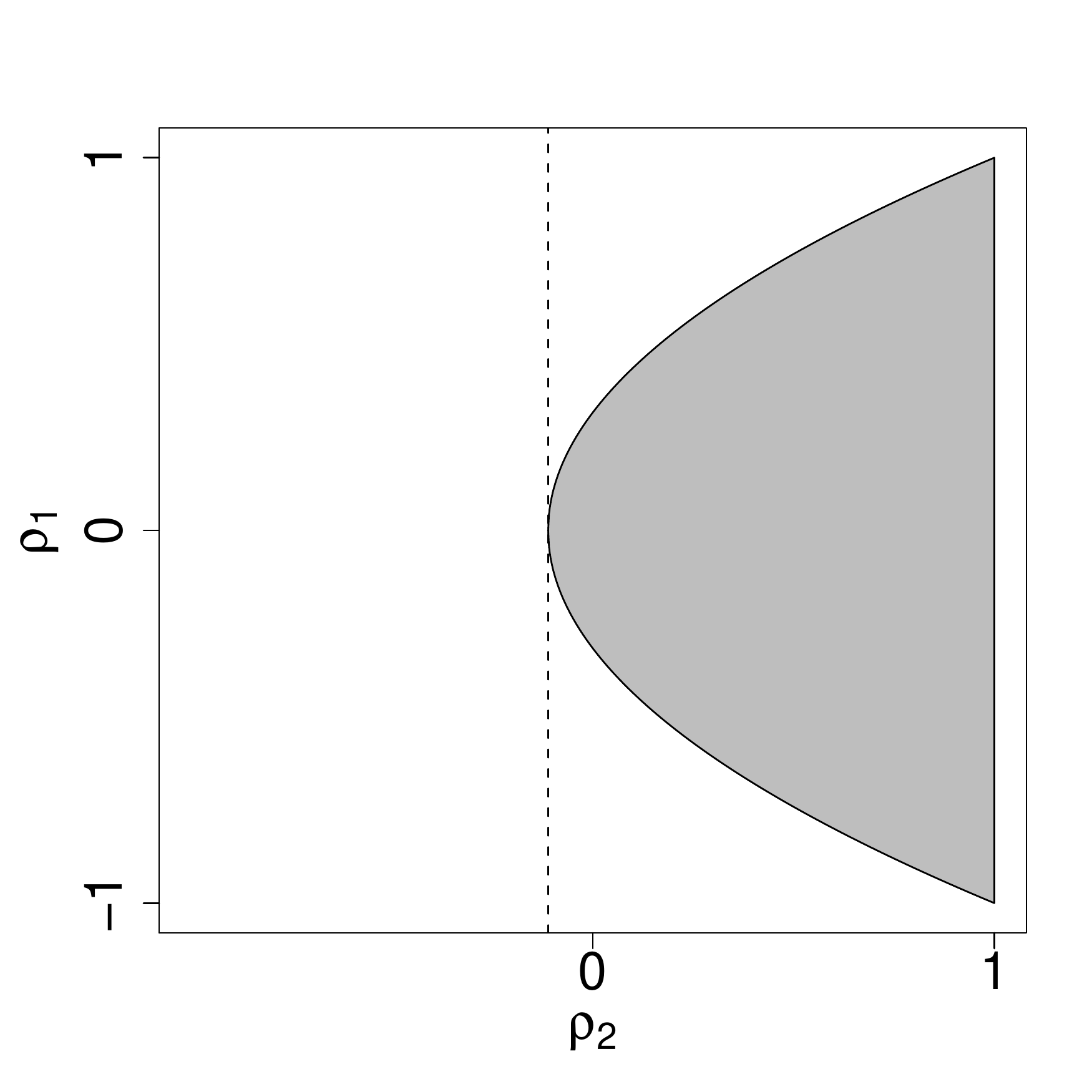}}
  \caption{Area of $(\rho_1, \rho_2)$ satisfying \eqref{eq.10} for different values of the $k$'s}\label{fig1}
\end{figure}

\subsection{Autoregressive correlation matrix}

We study a correlation matrix that has an autoregressive correlation structure for individual severities $(Y_1, \cdots, Y_N)$.
Based on this, we investigate an extended correlation matrix for $(N, Y_1, \cdots, Y_N)$.

\begin{definition}\label{def.2}
  For $\rho_1, \rho_2\in[-1,1]$, define the following matrices.
  \begin{enumerate}
       \item[i.] For any positive integer $k$, define a $k\times k$ matrix $\boldsymbol{\Sigma}_{\rho_2}^{[k, 2]}$ as
  \begin{equation*}
  \left[\boldsymbol{\Sigma}_{\rho_2}^{[k,2]}\right]_{i,j}:=\begin{cases}
    1, &\hbox{if}\quad i=j\\
    \rho_2^{\left\vert i-j\right\vert}, &\hbox{if}\quad i\neq j\\
  \end{cases}
  \end{equation*}
  for $i,j=1, \cdots, k$.
    \item[ii.] For any non-negative integer $k$, define a $(k+1)\times(k+1)$ matrix $\boldsymbol{\Sigma}_{\rho_1,\rho_2}^{[k,2]}$ as
  \begin{equation*}
    \boldsymbol{\Sigma}_{\rho_1,\rho_2}^{[k, 2]}:=\begin{cases}
    \left(
    \begin{array}{cc}
      1 & \rho_1 \left({\bf 1}_k\right)^{\mathrm T}\\
      \rho_1 {\bf 1}_k & \boldsymbol{\Sigma}_{\rho_2}^{[k,2]}\\
    \end{array}
    \right), & k=1,2, \cdots, ;\\
     1, & k=0;\\
    \end{cases}
  \end{equation*}
    where ${\bf 1}_k$ is a column vector of $1$ with length $k$.

  \end{enumerate}

\end{definition}

In matrix form, $\boldsymbol{\Sigma}_{\rho_2}^{[k,2]}$ and $\boldsymbol{\Sigma}_{\rho_1,\rho_2}^{[k, 2]}$ in Definition \ref{def.2} are written as
  \[
  \boldsymbol{\Sigma}_{\rho_2}^{[k,2]}=\left(
                      \begin{array}{ccccc}
                        1 & (\rho_2)^1 & (\rho_2)^2 & \cdots & (\rho_2)^{k-1} \\
                        (\rho_2)^1 & 1 & (\rho_2)^1 & \cdots & (\rho_2)^{k-2} \\
                        (\rho_2)^2 & (\rho_2)^1 & 1 & \cdots & (\rho_2)^{k-3} \\
                        \vdots & \vdots & \vdots & \ddots & \vdots \\
                        (\rho_2)^{k-1} & (\rho_2)^{k-2} & (\rho_2)^{k-3} & \cdots & 1 \\
                      \end{array}
                    \right)
  \]
and
\[
                      \boldsymbol{\Sigma}_{\rho_1,\rho_2}^{[k,2]}=\left(
                      \begin{array}{cccccc}
                        1 & \rho_1 & \rho_1 & \rho_1 & \cdots & \rho_1 \\
                        \rho_1 & 1 & (\rho_2)^1 & (\rho_2)^2 & \cdots & (\rho_2)^{k-1} \\
                        \rho_1 & (\rho_2)^1 & 1 & (\rho_2)^1 & \cdots & (\rho_2)^{k-2} \\
                        \rho_1 & (\rho_2)^2 & (\rho_2)^1 & 1 & \cdots & (\rho_2)^{k-3} \\
                        \vdots & \vdots & \vdots & \vdots & \ddots & \vdots \\
                        \rho_1 & (\rho_2)^{k-2} & (\rho_2)^{k-3} & \cdots & 1 \\
                      \end{array}
                    \right).
\]

The following proposition provides the determinants of $\boldsymbol{\Sigma}_{\rho_2}^{[k,2]}$ and $\boldsymbol{\Sigma}_{\rho_1,\rho_2}^{[k,2]}$.
\begin{proposition}\label{prop.oh.100}
For $\rho_1, \rho_2\in(-1,1)$ and any non-negative integer $k$, we have
  \begin{equation*}
  {\rm det}\left(\boldsymbol{\Sigma}_{\rho_2}^{[k,2]} \right)=
    (1-\rho_2^2)^{k}
  \end{equation*}
and
  \begin{equation*}
  {\rm det}\left(\boldsymbol{\Sigma}_{\rho_1,\rho_2}^{[k,2]} \right)=
1-\rho_1^2\left( k - 2\rho_2(k-1) + \rho_2^2 (k-2) \right).
  \end{equation*}
\end{proposition}

The proof is given in Appendix A.
We also provide the following well-known result without a proof.

\begin{proposition}\label{prop.oh.500}
If ${\rm det}\left(\boldsymbol{\Sigma}_{\rho_2}^{[k,2]} \right)\neq 0$, then
  \begin{equation*}
  (1-\rho_2^2)\left[\left( \boldsymbol{\Sigma}_{\rho_2}^{[k,2]} \right) ^{-1}\right]_{i,j}=
  \begin{cases}
    1, & (i,j) =(1,1) \quad\hbox{or}\quad (i,j) =(k,k);\\
    1+\rho_2^2, &   (i,j) \neq (1,1), \quad (i,j) \neq(k,k) \quad\hbox{and}\quad i=j;\\
    -\rho_2, & j=i+1 \quad\hbox{or}\quad j=i-1;\\
    0, & \hbox{otherwise}.
  \end{cases}
  \end{equation*}
\end{proposition}
Although the inverse matrix of $\boldsymbol{\Sigma}_{\rho_1,\rho_2}^{[k,2]}$
can be represented analytically using the Schur complement \citep{zhang2006schur}, we do not pursue it because
its representation is unnecessarily complicated.
By contrast, the condition for $\boldsymbol{\Sigma}_{\rho_1,\rho_2}^{[k,2]}$ to be positive definite is succinct,
as described in the following theorem. The proof is given in Appendix A.

\begin{theorem}\label{thm.100}
 Consider $\rho_1, \rho_2\in(-1,1)$.
Then, for a positive integer $k$,
  $\boldsymbol{\Sigma}_{\rho_2}^{[k,2]}$ is positive definite for any $\rho_1, \rho_2\in(-1,1)$.
  Furthermore,
for a non-negative integer $k$,
  $\boldsymbol{\Sigma}_{\rho_1,\rho_2}^{[k,2]}$ is positive definite for any $\rho_1, \rho_2\in(-1,1)$.
\end{theorem}

\section{Conditional Collective Risk Model}

This section presents the dependent collective risk model for positive frequency, which is called the ``conditional collective risk model'' throughout this paper. In Section \ref{sec.6}, we provide a generalized collective risk model that allows zero frequency.
Consider the positive frequency and severities:
\begin{equation}\label{eq.h.1}
N^+ \quad\hbox{and} \quad \left\{Y_1, Y_2, \cdots \right\}.
\end{equation}
Suppose that $n$ and $k$ are positive integers in $\mathbb{N}\times\mathbb{N}$.
The distribution function for positive frequency and severities is denoted by
\begin{equation}\label{eqq.o.1}
\P{N^+\le n, \boldsymbol{Y}^{[k]}\le \boldsymbol{y}^{[k]}}.
\end{equation}
$k$ in (\ref{eqq.o.1}) can be determined independently of $N^{+}$.
\citet{Marceau2018} studies a similar model, (\ref{eqq.o.1}), with $k=N^{+}$, but we allow $k$ to be any positive integer.
For example, we allow studying the distributions of
$\left(N^+, Y^{[1]}\right)=\left(N^+, Y_1\right)$
and $\left(N^+, Y^{[2]}\right)=\left(N^+, Y_1 , Y_2 \right)$
at the same time.
From the condition that the former should be obtained from the latter by integrating it with respect to $Y_2$,
in general, we require the distribution in (\ref{eqq.o.1}) to satisfy that for any $k_{1}<k_{2}$,
$\P{N^+\le n, \boldsymbol{Y}^{[k_{1}]}\le \boldsymbol{y}^{[k_{1}]}}$
is derived from $\P{N^+\le n, \boldsymbol{Y}^{[k_{2}]}\le \boldsymbol{y}^{[k_{2}]}}$ by taking the integration.

\begin{cond.model}\label{mod.0}
  Let $N^+\sim F_1^+$ be a non-degenerate positive integer-valued random variable, with the probability mass function $f_1$.
Then, we define the joint distribution of $(N^+, \boldsymbol{Y}^{[k]})$ as satisfying
      \begin{equation} \label{con.oh.2}
      (N^+, \boldsymbol{Y}^{[k]}) \sim
      H_{k}=C_{k+1}(F_1^+, F_2, \cdots, F_2), \\
      \end{equation}
for any positive integer $k$, where $F_2$ is non-negative continuous distribution that has $f_2$ as a probability density function. Here, $C_{k+1}$ is a $(k+1)$-dimensional copula satisfying the following inheritance property: for any $k_{1}<k_{2}$,
      \begin{equation} \label{eq.h.10}
C_{k_1+1}\left(u_1, \cdots, u_{k_1+1} \right)=
C_{k_2+1}\left(u_1, \cdots, u_{k_1+1}, 1, \cdots, 1 \right)\quad\hbox{for}\quad
u_1, \cdots, u_{k_1+1}\in[0,1].
\end{equation}

\end{cond.model}

For the two copulas $C_{k_{1}+1}$ and $C_{k_{2}+1}$ satisfying the inheritance property,
the corresponding distributions $H_{k_{1}+1}$ and $H_{k_{2}+1}$ also satisfy the inheritance property.
The copula models that we present in the following subsections satisfy such an inheritance property.
In contrast to ours, \citet{Marceau2018}'s model does not require condition \eqref{eq.h.10} to be satisfied.
In terms of parameter estimation, our model and \citet{Marceau2018}'s model, which have the same marginal distribution functions and copula functions, provide the same likelihood function because we always have $k=n^+$ in real observations.\footnote{Assume positive frequency for simplicity.}
However, to interpret and construct the dependence structure, allowing $k$ to be any positive integer as well as having the condition in \eqref{eq.h.10}, which is essentially the same as adding an assumption to the unobserved data, is critical, as commented in Remark \ref{rem.1}.

  The shared random effect model in \citet{Bastida}, \citet{Czado2015}, and \citet{PengAhn} and two-step frequency-severity model of \citet{Garrido}, \citet{park2018}, and \citet{AhnValdez2} require the distribution of individual severity to be in the EDF.
  Since these models require modeling the dependence between frequency and average severity (or aggregate severity), as mentioned by \citet{Garrido} and \citet{PengAhn}, the EDF assumption on individual severities is required to derive average severity. On the contrary, since our model does not require modeling the dependence between frequency and average severity, it allows $F_2$ to be any distribution function, including a heavy-tailed distribution function.

Similar to \citet{Marceau2018}, for a positive integer $k$, the joint density function in Conditional Model \ref{mod.0} can be written as
\begin{equation}\label{eq.61}
\begin{aligned}
&\frac{\partial^k}{\partial y_1\cdots\partial y_k}\left( H_k\left(n, \boldsymbol{y}^{[k]}\right)
- H_k\left(n-1, \boldsymbol{y}^{[k]}\right)  \right)\\
&\quad\quad\quad\quad\quad= \left( \P{N^+\le n\big\vert \boldsymbol{y}^{[k]}}-\P{N^+\le n-1\big\vert \boldsymbol{y}^{[k]}}
\right)  f_{\boldsymbol{Y}^{[k]}}(\boldsymbol{y}^{[k]})
\end{aligned}
\end{equation}
      for $n\in \mathbb{Z}$ and $\boldsymbol{y}^{[k]}\in\Real^k$,
      where
      $f_{\boldsymbol{Y}^{[k]}}:\Real_+^k\mapsto \Real_+$ is the probability density function of $\boldsymbol{Y}^{[k]}$.
The computational complexity of \eqref{eq.61} depends on the choice of the $(k+1)$-dimensional copula $C_{k+1}$. In particular,
the conditional distribution of $N^+$ in \eqref{eq.61} is involved in the integration of the copula, which increases the complexity of the estimation procedure.
In the cases of the Gaussian copula and t-copula, the conditional distribution parts in \eqref{eq.61} are relatively easy to compute
as well as provide an intuitive interpretation of the dependence structure using the form of the covariance matrix.
In this section, we present the Gaussian copula and t-copula versions of Conditional Model \ref{mod.0} in detail.

\subsection{Gaussian copula model}
In the following, we provide the distribution of $(N^+, \boldsymbol{Y}^{[k]})$ for a positive integer $k$ based on the Gaussian copula family with the equicorrelation matrix, $\boldsymbol{\Sigma}_{\rho_1, \rho_2}^{[k, z]}$ for $z=1$, and with the autoregressive correlation matrix, $\boldsymbol{\Sigma}_{\rho_1, \rho_2}^{[k, z]}$ for $z=2$.

\begin{cond.model}\label{mod.1}
  Let $N^+\sim F_1^+$ be a non-degenerate positive integer-valued random variable, with the probability mass function $f_1$.
Consider the positive definite matrices
\[
\boldsymbol{\Sigma}_{\rho_2}^{[k, z]}\quad\hbox{and}\quad \boldsymbol{\Sigma}_{\rho_1, \rho_2}^{[k, z]}
\]
for $z=1,2$.
We assume
\begin{equation} \label{con.oh.21}
(\rho_1, \rho_2)\in\left\{ (\rho_1, \rho_2)\in(-1,1)^2 \big\vert  \rho_1^2<\rho_2<1 \right\}
\end{equation}
for $z=1$, and assume
\begin{equation} \label{c.o.1}
(\rho_1, \rho_2)\in (-1,1)^2
\end{equation}
for $z=2$.
Then, for the given correlation matrix structure $z\in\{1,2\}$, we define the joint distribution of $(N^+, \boldsymbol{Y}^{[k]})$ as satisfying
      \begin{equation} \label{con.oh.21}
      (N^+, \boldsymbol{Y}^{[k]}) \sim
      H_{N^+, \boldsymbol{Y}^{[k]}}=C_{\boldsymbol{\Sigma}_{\rho_1, \rho_2}^{[k, z]}}(F_1^+, F_2, \cdots, F_2), \\
      \end{equation}
for any positive integer $k$, where $F_2$ is a non-negative continuous distribution that has $f_2$ as a probability density function. Here, $C_{\boldsymbol{\Sigma}_{\rho_1, \rho_2}^{[k, z]}}$ is a Gaussian copula, with a corresponding density function denoted as $c_{\boldsymbol{\Sigma}_{\rho_1, \rho_2}^{[k, z]}}$ and correlation matrix $\boldsymbol{\Sigma}_{\rho_1, \rho_2}^{[k, z]}$.

\end{cond.model}

It is straightforward to check that the copula
$$C_{\boldsymbol{\Sigma}_{\rho_1, \rho_2}^{[k,z]}}, \quad\hbox{for}\quad z=1, 2$$
satisfies the condition in \eqref{eq.h.10}. The following lemma shows the explicit form of the density function of Conditional Model \ref{mod.1}.

\begin{lemma}\label{prop.2}

Considering the frequency and severities in Conditional Model \ref{mod.1}, for each $z=1, 2$, we have the following results.

      \begin{enumerate}
        \item[i.] For a positive integer $k$, the joint density function is given by
      \begin{equation}\label{eq.14}
      \begin{aligned}
      &h_{N^+, \boldsymbol{Y}^{[k]}}(n, \boldsymbol{y}^{[k]})\\
      &\quad=
      f_{\boldsymbol{Y}^{[k]}}(\boldsymbol{y}^{[k]})
      \left(
      \Phi\left( \frac{\Phi^{-1}(F_1^+(n)) -\mu_{[k, z]}}{\sigma_{[k, z]}} \right) - \Phi\left( \frac{\Phi^{-1}(F_1^+(n-1)) -\mu_{[k, z]}}{\sigma_{[k, z]}} \right)
      \right)\\
      \end{aligned}
      \end{equation}
      for $n\in \mathbb{Z}$ and $\boldsymbol{y}^{[k]}\in\Real^k$,
      where
      $f_{\boldsymbol{Y}^{[k]}}$ is a density function given as
      \begin{equation}\label{eq.oh.6}
      f_{\boldsymbol{Y}^{[k]}}(y_1, \cdots, y_k)=c_{\boldsymbol{\Sigma}_{\rho_2}^{[k, z]}}(F_2(y_1), \cdots, F_2(y_k))\prod_{i=1}^{k}f_2(y_i),
      \end{equation}
       which is the probability density function of the cumulative distribution function $C_{\boldsymbol{\Sigma}_{\rho_2}^{[k, z]}}(F_2, \cdots, F_2)$. Here, $\mu_{[k, z]}$ and $\sigma_{[k, z]}$ are defined as
      \begin{equation}\label{eq.18}
       \mu_{[k, z]}:=\left( \rho_1\boldsymbol{1}_k^{\mathrm T}\right) \left(\boldsymbol{\Sigma}_{\rho_2}^{[k, z]}\right)^{-1} \left(\Phi^{-1}(F_2(y_1)), \cdots, \Phi^{-1}(F_2(y_k)) \right)^\mathrm{T}
\end{equation}
and
      \begin{equation}\label{eq.18.a}
         \sigma_{[k, z]}:=\sqrt{1- \left( \rho_1\boldsymbol{1}_k^{\mathrm T}\right) \left(\boldsymbol{\Sigma}_{\rho_2}^{[k, z]}\right)^{-1} \left( \rho_1\boldsymbol{1}_k\right)}.
\end{equation}

\item[ii.] For a positive integer $k$, the conditional density function of is given $\boldsymbol{Y}^{[k]}$ for given $N^+=n$ is
\begin{equation}\label{eq.567}
\begin{aligned}
  &h_{\boldsymbol{Y}^{[k]}\big\vert N^+}\left(\boldsymbol{Y}^{[k]}\big\vert n \right)\\
      &\quad\quad\quad=
      \frac{f_{\boldsymbol{Y}^{[k]}}(\boldsymbol{y}^{[k]})}{f_1^+(n)}
      \left(
      \Phi\left( \frac{\Phi^{-1}(F_1^+(n)) -\mu_{[k, z]}}{\sigma_{[k, z]}} \right) - \Phi\left( \frac{\Phi^{-1}(F_1^+(n-1)) -\mu_{[k, z]}}{\sigma_{[k, z]}} \right)
      \right).
\end{aligned}
\end{equation}
      \end{enumerate}
\end{lemma}

The proof in given in Appendix B.
Conditional Model \ref{mod.1} has an advantage when investigating the degree of dependence among the variables via Spearman's rho.
The definition of Spearman's rho is first given below.
\begin{definition}
  Define Spearman's rho of $(V_1, V_2)\sim C(G_1, G_2)$ for some bivariate copula $C$ and the marginal distributions $G_1$ and $G_2$ as
  \[
  \rho(V_1, V_2)=12\int\int C(G_1(v_1), G_2(v_2)) -G_1(v_1)G_2(v_2) {\rm d}v_1  {\rm d}v_2.
  \]
\end{definition}
For the details on Spearman's rho, see \citet{Nelson}.
Spearman's rho in Conditional Model \ref{mod.1} can be obtained from the following results.

\begin{corollary}\label{cor.3}
  Consider the frequency and severities defined in Conditional Model \ref{mod.1}. Then, for each $z=1, 2$, we have the following results.
    \begin{enumerate}
      \item[i.] Spearman's rho between $N^+$ and $Y_j$ can be calculated as
      \[
      \rho(N^+, Y_j)=\sum\limits_{n=0}^{\infty}\int_{0}^{\infty} 12\left[C_{\rho_1}(F_1^+(n), F_2(y)) - F_1^+(n) F_2(y) \right] f_1(n)f_2(y) {\rm d} y
      \],
     where $C_{\rho_1}$ is the Gaussian copula with the correlation coefficient $\rho_1$.

      \item[ii.] Spearman's rho between $Y_{j_1}$ and $Y_{j_2}$ can be calculated as
      $$\rho(Y_{j_1}, Y_{j_2})=\frac{6}{\pi}\arcsin\left( \frac{\rho_2}{2}\right), \quad\hbox{for}\quad z=1$$
      and
      $$\rho(Y_{j_1}, Y_{j_2})=\frac{6}{\pi}\arcsin\left( \frac{\rho_2^{\left\vert k_1-k_2 \right\vert}}{2}\right), \quad\hbox{for}\quad z=2$$
      for the positive integers $j_1$ and $j_2$ satisfying $j_1\neq j_2$.
    \end{enumerate}
\end{corollary}
The proof of the first part is immediate from the definition, and the proof of the second part can be found in \citet{Kruskal1958}.

\begin{remark}\label{rem.1}
  In Conditional Model \ref{mod.0}, allowing $k$ to be any integer value makes the definition of the dependence measures such as
  \begin{equation}\label{eq.101}
  \rho\left(N^+, Y_j \right) \quad \hbox{and}\quad \rho\left(Y_{j_1}, Y_{j_2}\right), \quad j_1\neq j_2
  \end{equation}
  in Corollary \ref{cor.3} well defined and interpreted straightforwardly. On the contrary, fixing $k= N^{+}$
as in \citet{Marceau2018} may complicate the interpretation.\footnote{The main purpose of \citet{Marceau2018} is the construction of the collective risk model based on a hierarchical Archimedean copula, which does not require the risk measure such as \eqref{eq.101}.} Indeed, under the model with fixed $k=N^+$ only,
the definition of the (marginal) dependence measures in \eqref{eq.101} may be loose because $Y_{j}$ is defined only
when $N^{+}\geq j$. Therefore,
the model with fixed $k=N^+$ only is suitable to discuss the following conditional versions of the dependence measures
  \begin{equation}\label{eq.1001}
    \rho\left(N^+, Y_j \big\vert N^+\ge j\right) \quad \hbox{and}\quad \rho\left(Y_{j_1}, Y_{j_2}\big\vert N^+\ge \max\left\{j_1, j_2 \right\}\right), \quad j_1\neq j_2.
 \end{equation}
However, interpreting the dependence structure with such dependence measures \eqref{eq.1001} can be difficult. As a consequence, the interpretation of the (marginal) correlation matrix
\[
\boldsymbol{\Sigma}_{\rho_1, \rho_2}^{[k, z]}, \quad z=1, 2
\]
  introduced in Section \ref{sec.4} may be difficult under Conditional Model \ref{mod.0} with fixed $k=N^+$ only.
\end{remark}

\subsection{T-copula model}

The Gaussian copula is not an inevitable choice for modeling frequency and severities. However, for the convenience of the statistical modeling and estimation, we focus on the copula family specified by the covariance (correlation) matrix only.
The elliptical copula family has such a property. Among them, we focus on the t-copula that has gained broad popularity in recent studies.

We first define the $k$-dimensional multivariate t-distribution $\boldsymbol{Z}\sim {\rm MVT}(\boldsymbol{0}_k, \boldsymbol{\Sigma}, {\rm df})$, with scale matrix $\boldsymbol{\Sigma}$, and degrees of freedom $\rm df$, which has the following probability density function:
\[
f({\boldsymbol z}) = \frac{\Gamma\left(\frac{k+{\rm df}}{2}\right)}{\Gamma\left( \frac{\rm df}{2}\right) (\pi {\rm df})^{k/2} \left\vert \boldsymbol{\Sigma}\right\vert^{1/2}}
\left( 1+ \frac{1}{\rm df}\boldsymbol{k}^{\mathrm T}\boldsymbol{\Sigma}^{-1}\boldsymbol{k}\right)
\]
for $\boldsymbol{z}=(z_1, \cdots, z_k)\in\Real^k$.
We denote the corresponding t-copula as $C_{{\rm df}, \boldsymbol{\Sigma}_{\rho_1, \rho_2}^{[k,1]}}$.
We also represent $\Phi_{\rm df}$ and $\phi_{\rm df}$ by the cumulative density function and probability density function of the univariate Student's t-distribution, respectively.
The t-copula version of Conditional Model \ref{mod.0} is provided below.

\begin{cond.model}\label{mod.10}
  Let $N^+\sim F_1^+$ be a positive integer-valued random variable, with the probability mass function $f_1^+$,
and assume \eqref{con.oh.21} and \eqref{c.o.1} for $z=1$ and $z=2$, respectively. Then, for each $z=1,2$, define
the joint distribution of $(N, \boldsymbol{Y}^{[k]})$ as in
      \begin{equation*} 
      (N^+, \boldsymbol{Y}^{[k]}) \sim
      C_{{\rm df}, \boldsymbol{\Sigma}_{\rho_1, \rho_2}^{[k, z]}}(F_1^+, F_2, \cdots, F_2), \\
      \end{equation*}
for any positive integer $k$, where $F_2$ is a non-negative continuous distribution that has $f_2$ as a probability density function. Here, $C_{{\rm df},\boldsymbol{\Sigma}_{\rho_1, \rho_2}^{[k, z]}}$ is a t-copula, which has the corresponding density function denoted as $c_{{\rm df},\boldsymbol{\Sigma}_{\rho_1, \rho_2}^{[k, z]}}$, with scale matrix $\boldsymbol{\Sigma}_{\rho_1, \rho_2}^{[k, z]}$ and degree of freedom $\rm df$.

\end{cond.model}

The following result is the t-copula version of Lemma \ref{prop.2}.
\begin{lemma}\label{l.2}
Considering the frequency and severities in Conditional Model \ref{mod.10} with $z=1$ and $z=2$, we have the following results.

      \begin{enumerate}
        \item[i.] For a positive integer $k$, the joint density function is given by
      \begin{equation}\label{eq.141}
      \begin{aligned}
      &h_{N^+, \boldsymbol{Y}^{[k]}}(n, \boldsymbol{y}^{[k]})\\
      &\quad=
      f_{\boldsymbol{Y}^{[k]}}^{*}(\boldsymbol{y}^{[k]})
      \left(
      \Phi_{{\rm df}+k}\left( \frac{\Phi_{\rm df}^{-1}(F_1^+(n)) -\mu_{[k, z]}^*}{\sigma_{[k, z]}^*} \right) - \Phi_{{\rm df}+k}\left( \frac{\Phi_{\rm df}^{-1}(F_1^+(n-1)) -\mu_{[k, z]}^*}{\sigma_{[k, z]}^*} \right)
      \right)
      \end{aligned}
      \end{equation}
      for $n\in \mathbb{N}^+$ and $\boldsymbol{y}^{[k]}\in\Real^k$,
      where
      $f_{\boldsymbol{Y}^{[k]}}^{*}(\boldsymbol{y}^{[k]})$ is a density function given as
      \begin{equation*}
      f_{\boldsymbol{Y}^{[k]}}^{*}(y_1, \cdots, y_k)=c_{{\rm df}, \boldsymbol{\Sigma}_{\rho_2}^{[k, z]}}(F_2(y_1), \cdots, F_2(y_k))\prod_{i=1}^{k}f_2(y_i),
      \end{equation*}
       which is the probability density function of the cumulative distribution function $C_{{\rm df},\boldsymbol{\Sigma}_{\rho_2}^{[k, z]}}(F_2, \cdots, F_2)$.
Here, $\mu_{[k, z]}^*$ and $\sigma_{[k, z]}^*$ are defined as
      \begin{equation}\label{e.18}
       \mu_{[k, z]}^*:=\left( \rho_1\boldsymbol{1}_k^{\mathrm T}\right) \left(\boldsymbol{\Sigma}_{\rho_2}^{[k, z]}\right)^{-1} \left(\Phi_{\rm df}^{-1}(F_2(y_1)), \cdots, \Phi_{\rm df}^{-1}(F_2(y_k)) \right)^\mathrm{T}
\end{equation}
and
      \begin{equation}\label{e.19}
      \begin{aligned}
         &\sigma_{[k, z]}^*:=\sigma_{[k, z]}\\
         &\quad\times \sqrt{\frac{
         {\rm df}+
         \left(\Phi_{\rm df}^{-1}(F_2(y_1)), \cdots, \Phi_{\rm df}^{-1}(F_2(y_k)) \right)
         \left(\boldsymbol{\Sigma}_{\rho_2}^{[k, z]}\right)^{-1}
         \left(\Phi_{\rm df}^{-1}(F_2(y_1)), \cdots, \Phi_{\rm df}^{-1}(F_2(y_k)) \right)^\mathrm{T}
         }{{\rm df}+k}
         }.
         \end{aligned}
\end{equation}

\item[ii.] For a positive integer $k$, the conditional density function of is given $\boldsymbol{Y}^{[k]}$ for given $N^+=n$ is
\begin{equation}\label{eq.567}
\begin{aligned}
  &h_{\boldsymbol{Y}^{[k]}\big\vert N^+}\left(\boldsymbol{Y}^{[k]}\big\vert n \right)\\
      &\quad\quad\quad=
            \frac{f_{\boldsymbol{Y}^{[k]}}^*(\boldsymbol{y}^{[k]})}{f_1^+(n)}
      \left(
      \Phi_{{\rm df}+k}\left( \frac{\Phi_{\rm df}^{-1}(F_1^+(n)) -\mu_{[k, z]}^*}{\sigma_{[k, z]}^*} \right) - \Phi_{{\rm df}+k}\left( \frac{\Phi_{\rm df}^{-1}(F_1^+(n-1)) -\mu_{[k, z]}^*}{\sigma_{[k, z]}^*} \right)
      \right).
\end{aligned}
\end{equation}
      \end{enumerate}

\end{lemma}
For brevisity, we omit the proof because it is similar to that of Lemma \ref{prop.2}.

\section{Collective Risk Model with the Observed Data}\label{sec.6}

Conditional Models \ref{mod.1} and \ref{mod.10} cannot be directly applied to real data including zero frequency because they assume that the frequency is positive.
In the following, we explain how to modify Conditional Models \ref{mod.1} and \ref{mod.10} to accommodate zero frequency.

\begin{model}[Dependent collective risk model]\label{mod.2}
  For each correlation matrix $z=1,2$, consider the frequency and severities defined in Conditional Model \ref{mod.1} or \ref{mod.10}.
      We assume that $R$ is a Bernoulli random variable with success probability $p$,
and it is mutually independent of $(N^+, Y_1, \cdots, Y_{k})$ for any $k>0$. Define
\[
N:=\begin{cases}
  N^{+}, & R=1;\\
  0, & R=0;\\
\end{cases}
\]
      and its cumulative distribution function and density function as $F_1$ and $f_1$, respectively.
      Then, define
\[
(N,\boldsymbol{Y}^{[N]}):=\begin{cases}
  (N, Y_1, \cdots, Y_N), & N\ge 1;\\
  0, & N=0;\\
\end{cases}
\]
as a dependent collective risk model.
\end{model}

Model \ref{mod.2} is called the ``dependent collective risk model'' throughout the paper.
The following theorems, which are the corollaries of Lemmas \ref{prop.2} and \ref{l.2}, show the joint density function of $(N, \boldsymbol{Y}^{[N]})$ in Model \ref{mod.2}.

\begin{theorem}\label{thm.2}
Consider Model \ref{mod.2} along with Conditional Model \ref{mod.1}.
Then, for each $z=1, 2$, the joint density function of the discrete margin $N$ and continuous margins $\boldsymbol{Y}_N$ is given by
     \begin{equation}\label{eq.17}
     \begin{aligned}
      &h(n, \boldsymbol{y}^{[n]})\\
      &\quad=
      \begin{cases}
      p\, f_{\boldsymbol{Y}^{[n]}}(\boldsymbol{y}^{[n]})
      \left(
      \Phi\left( \frac{\Phi^{-1}(F_1^+(n)) -\mu_{[n,z]}}{\sigma_{[n,z]}} \right) - \Phi\left( \frac{\Phi^{-1}(F_1^+(n-1)) -\mu_{[n,z]}}{\sigma_{[n,z]}} \right)
      \right), & n\in\mathbb{N},\quad \boldsymbol{y}^{[n]}\in\Real^n\\
      1-p, &n=0,\\
      \end{cases}
      \end{aligned}
      \end{equation}
       where $\mu_{[n,z]}$ and $\sigma_{[n,z]}$ are defined in \eqref{eq.18} and \eqref{eq.18.a}.

\end{theorem}

The proof is provided here. \eqref{eq.17} is trivial for $n=0$. Now, consider the case for $n>0$.
    For an observation $(n, y_1, \cdots, y_n)$, the corresponding likelihood function is
    \[
    \begin{aligned}
      h(n, \boldsymbol{y}^{[n]})&=p\, f_1^+(n)h_{\boldsymbol{Y}^{[n]}\big\vert N^+}(y_1, \cdots, y_n\big\vert n)\\
      &=p\, h_{N^+, \boldsymbol{Y}^{[n]}}(n, \boldsymbol{y}^{[n]}),
      \end{aligned}
    \]
    where the second equality is from part ii of Lemma \ref{prop.2}.

\begin{theorem}\label{th.3}
Consider Model \ref{mod.2} along with Conditional Model \ref{mod.10}.
Then, for each $z=1, 2$, the joint density function of the discrete margin $N$ and continuous margins $\boldsymbol{Y}_N$ is given by
     \begin{equation}\label{eq.17_2}
     \begin{aligned}
      &h(n, \boldsymbol{y}^{[n]})\\
      &\quad=
      \begin{cases}
      p\, f_{\boldsymbol{Y}^{[n]}}(\boldsymbol{y}^{[n]})
      \left(
      \Phi_{{\rm df}+n}\left( \frac{\Phi_{\rm df}^{-1}(F_1^+(n)) -\mu_{[n,z]}^*}{\sigma_{[n,z]}^*} \right) - \Phi_{{\rm df}+n}\left( \frac{\Phi_{\rm df}^{-1}(F_1^+(n-1)) -\mu_{[n,z]}^*}{\sigma_{[n,z]}^*} \right)
      \right), & n\in\mathbb{N},\quad \boldsymbol{y}^{[n]}\in\Real^n\\
      1-p, &n=0,\\
      \end{cases}
      \end{aligned}
      \end{equation}
       where $\mu_{[n,z]}^*$ and $\sigma_{[n,z]}^*$ are defined in \eqref{e.18} and \eqref{e.19}.
\end{theorem}

For brevity, we omit the proof because it is the same as in Theorem \ref{thm.2}.

\subsection{Derivation of some useful quantities}
Using Lemma \ref{cal.lem} in the Appendix B, the following proposition shows how to derive some useful quantities from Model \ref{mod.2}.
\begin{proposition}
     Consider Model \ref{mod.2} along with either Conditional Model \ref{mod.1} or Conditional Model \ref{mod.10}.
Then, we have
\[
\E{S} = \E{\sum_{j=1}^{N}Y_{j}}=p \sum\limits_{n=1}^{\infty} n \int_0^\infty y\, h_{N^+, \boldsymbol{Y}^{[1]}}(n, y) {\rm d} y
\]
and
\[
\cov{N, S} = p\sum\limits_{n=1}^{\infty} n^2 \int_0^\infty y
  h_{N^+, \boldsymbol{Y}^{[1]}}(n, y) {\rm d} y
  -\left(p \sum\limits_{n=1}^{\infty} n f_{N^+}(n)\right)
    \left(p \sum\limits_{n=1}^{\infty} n \int_0^\infty y\, h_{N^+, \boldsymbol{Y}^{[1]}}(n, y) {\rm d} y
    \right),
\]
where $h_{N^+, \boldsymbol{Y}^{[1]}}$ is defined in \eqref{eq.14} or \eqref{eq.141}, depending on the assumptions of Conditional Model \ref{mod.1} or Conditional Model \ref{mod.10} as well as the type of covariance matrix.
Furthermore, we have
\[
\begin{aligned}
\cov{N, M\big\vert N>0} &=\sum\limits_{n=1}^{\infty} n \int_0^\infty y
  h_{N^+, \boldsymbol{Y}^{[1]}}(n, y) {\rm d} y-
  \left( \sum\limits_{n=1}^{\infty} n f_{N^+}(n) \right)
  \left( \sum\limits_{n=1}^{\infty} n \int_0^\infty y
  h_{N^+, \boldsymbol{Y}^{[1]}}(n, y) {\rm d} y \right).\\
 \end{aligned}
\]
\end{proposition}

The detailed derivation steps of Proposition 5 are given in Appendix C.

\subsection{Extension to regression models}

We provide regression models for the dependent collective risk model below.

\begin{model}\label{mod.3}

For each individual $i=1, \cdots, I$, consider the dependent collective risk model for
\[
(N_i,\boldsymbol{Y}_i^{[N_i]}):=\begin{cases}
  (N_i, Y_{i1}, \cdots, Y_{iN_i}), & N_i\ge 1;\\
  0, & N_i=0;\\
\end{cases}
\]
 in Model \ref{mod.2}.
  Let
  $\hbox{
  $\left(\boldsymbol{x}_i,\boldsymbol{x}_i^*,  \boldsymbol{w}_{i}\right)$}
  $ be the given characteristics of the $i$-th policyholder.
  Consider the following model.
  \begin{enumerate}
    \item[i.] For the frequency part, use the hurdle regression model with $N_{it}$. Specifically, use $\boldsymbol{x}_i$ and $\boldsymbol{x}_i^*$ as the explanatory variables for the zero and positive frequency parts, respectively. Denote $\boldsymbol{\beta}$ and $\boldsymbol{\beta}^*$ as the corresponding sets of regression coefficients. Specifically,
        \[
        N_{i}^+ \sim {F}_1^+\left(\cdot ; \lambda_{i}, \psi_1 \right), \quad\hbox{with}\quad \eta_1(\lambda_{i})=\boldsymbol{x}_i\boldsymbol{\beta}
        \]
        and
        \[
        R_{i} \sim {\rm Ber}\left(p_i \right), \quad\hbox{with}\quad \eta_1^*(p_i)=\boldsymbol{x}_i^*\boldsymbol{\beta}^*
        \]
    for some link functions $\eta_1$ and $\eta_1^*$. 

   \item[ii.] For the severity part, use the regression model with $Y_{it}$ and $\boldsymbol{w}_{i}$ as the dependent variable and the set of explanatory variables, respectively. Denote $\boldsymbol{\gamma}$ as the corresponding sets of regression coefficients. Specifically, for $N_i>0$,
  \begin{equation*}
       Y_{it} \sim {F}_2\left(\cdot\,;\,  \xi_{i}, \psi_2\right), \quad\hbox{with}\quad \eta_2(\xi_{i})=\boldsymbol{w}_{i} \boldsymbol{\gamma}
  \end{equation*}
   for $t\in\mathcal{N}$ with some link function $\eta_2$. 
  \end{enumerate}

\end{model}

Model \ref{mod.3} is a flexible regression model for frequency and severities.
Its marginal distribution of frequency is given as
\[
F_1(n; \lambda_{i}, \psi_1) = 1-p_i +p_i F_1^+(n; \lambda_{i}, \psi_1)
\]
for $n=0,1,\ldots$.
A special case of Model \ref{mod.3} is particularly interesting because it is convenient for statistical estimation.
Consider any distribution function $F_1$ that can explain the frequency part including zero.
Define \[
p_i=  1 - {F}_1\left(0 ; \lambda_{i}, \psi_1 \right)
\]
and
\[
F_1^+(n; \lambda_i, \psi_1)=\frac{{F}_1\left(n ; \lambda_{i}, \psi_1 \right)-{F}_1\left(0 ; \lambda_{i}, \psi_1 \right)}{1-{F}_1\left(0 ; \lambda_{i}, \psi_1 \right)}
\]
for $n\in \mathbb{N}_+$. Then, we have
\[
N_i \sim F_1(\cdot; \lambda_i, \psi_1).
\]

We use this model for the simulation and data analysis in the following sections.
Specifically, we use the Poisson distribution with mean $\lambda_i=\exp(\boldsymbol{x}_{i}\boldsymbol{\beta})$ for $N_i$, and
Gamma distribution $\xi_i=\exp(\boldsymbol{w}_{i}\boldsymbol{\gamma})$ for $Y_{it}$. 

\section{Numerical Study}

We conduct a simulation study to investigate the finite sample properties of the parameter estimates and effect of the dependence between frequency and severity on them for the proposed model. The portfolio of policyholders of size $I=5000$ are generated from the proposed model under 12 scenarios motivated by the real data analysis in Section \ref{sec.8}.
Table \ref{sim_param} provides the details of the parameter settings.
In each simulation, two predictors $\boldsymbol{x}_{i}$ and $\boldsymbol{w}_{i}$ are used and generated from Bernoulli(0.5) independently.

\begin{table}[h!]
\caption{Parameter settings for the 12 scenarios}
\centering
\begin{tabular}{ l r r r r r r r r r}
 \hline
 &\multicolumn{5}{c}{Parameter} \\
 \cline{2-10}
 Scenario & $\beta_0$ & $\beta_1$ & $\beta_2$ &  $\gamma_0$ & $\gamma_1$ & $\gamma_2$ & $\nu$ & $\rho_1$ & $\rho_2$\\
 \hline
1 	& -2.5 & 0.5 & 1.0 & 8	& -0.1&	0.3& 0.7& -0.05 & 0.10 \\
2 	& -2.5 & 0.5 & 1.0 & 8	& -0.1&	0.3& 0.7& -0.05 & 0.05 \\
3 	& -2.5 & 0.5 & 1.0 & 8	& -0.1&	0.3& 0.7& 0.05 & 0.10 \\
4 	& -2.5 & 0.5 & 1.0 & 8	& -0.1&	0.3& 0.7& 0.05 & 0.05 \\
5 	& -2.5 & 0.5 & 1.0 & 8	& -0.1&	0.3& 0.7& 0.10 & 0.10 \\
6 	& -2.5 & 0.5 & 1.0 & 8	& -0.1&	0.3& 0.7& 0.10 & 0.05 \\
7 	& -2.5 & 0.5 & 1.5 & 8	& -0.1&	0.3& 0.7& -0.05 & 0.10 \\
8 	& -2.5 & 0.5 & 1.5 & 8	& -0.1&	0.3& 0.7& -0.05 & 0.05 \\
9 	& -2.5 & 0.5 & 1.5 & 8	& -0.1&	0.3& 0.7& 0.05 & 0.10 \\
10 	& -2.5 & 0.5 & 1.5 & 8	& -0.1&	0.3& 0.7& 0.05 & 0.05 \\
11	& -2.5 & 0.5 & 1.5 & 8	& -0.1&	0.3& 0.7& 0.10 & 0.10 \\
12	& -2.5 & 0.5 & 1.5 & 8	& -0.1&	0.3& 0.7& 0.10 & 0.05 \\
\hline
\end{tabular}
\label{sim_param}
\end{table}

For each scenario, Table \ref{sim_rb} and Table \ref{sim_mae} summarize the simulation results from 500 independent Monte Carlo samples, including the relative bias and mean squared error (MSE) of the parameter estimates.
Table \ref{sim_rb} indicates that in all the scenarios, the estimates are close to the true parameters of the proposed model and shows that the relative bias and MSE are small.
A relative bias larger than 10$\%$ is only observed for $\rho_2$ in scenario 5, which has relatively high correlations for $\rho_1$ and $\rho_2$.

\begin{table}[h!]
\caption{Relative bias in $\%$ for all the parameters from the 12 scenarios}
\centering
\begin{tabular}{ l r r r r r r r r r   }
 \hline
 &\multicolumn{5}{c}{Relative Bias ($\%$)} \\
 \cline{2-10}
 Scenario &  $\beta_0$ & $\beta_1$ & $\beta_2$ & $\gamma_0$ & $\gamma_1$ & $\gamma_2$ & $\nu$& $\rho_1$ & $\rho_2$\\
 \hline
1&	0.06 	&	0.34 	&	0.01 	&	0.01 	&	-1.57 	&	-0.10 	&	0.19 	&	5.79 	&	8.61	\\
2&	0.10 	&	0.35 	&	0.11 	&	0.11 	&	-0.55 	&	-0.59 	&	0.16 	&	0.36 	&	0.42	\\
3&	-0.12 	&	-0.57 	&	-0.17 	&	-0.17 	&	-0.60 	&	-0.73 	&	0.23 	&	1.74 	&	6.77	\\
4&	0.01 	&	-0.36 	&	0.11 	&	0.11 	&	-4.00 	&	0.46 	&	0.43 	&	0.94 	&	-2.63\\
5&	0.00 	&	0.32 	&	-0.05 	&	-0.05 	&	0.89 	&	-0.77 	&	0.11 	&	1.58 	&	14.90\\
6&	0.02 	&	-0.13 	&	-0.07 	&	-0.07 	&	-1.62 	&	-1.50 	&	0.31 	&	-1.45 	&	-0.83\\
7&	0.01 	&	-0.07 	&	-0.03 	&	-0.03 	&	2.98 	&	1.15 	&	0.12 	&	0.66 	&	-0.08\\
8&	0.03 	&	0.10 	&	-0.08 	&	-0.08 	&	-2.36 	&	0.05 	&	0.11 	&	-1.92 	&	-2.80\\
9&	0.03 	&	0.43 	&	0.04 	&	0.04 	&	0.68 	&	1.50 	&	0.26 	&	1.14 	&	4.61 \\
10&	0.16 	&	0.07 	&	0.22 	&	0.22 	&	-1.39 	&	-0.72 	&	0.04 	&	0.98 	&	-0.73\\
11&	-0.10 	&	-0.30 	&	-0.15 	&	-0.15 	&	-0.95 	&	0.40 	&	0.08 	&	1.89 	&	3.16 \\
12&	-0.04 	&	-0.03 	&	-0.18 	&	-0.18 	&	-0.80 	&	0.88 	&	0.28 	&	0.68 	&	-2.55 \\
\hline
\end{tabular}
\label{sim_rb}
\end{table}

\begin{table}[h!]
\caption{Mean absolute error for all the parameters from the 12 scenarios}
\centering
\begin{tabular}{ l r r r r r r r r r   }
 \hline
 &\multicolumn{5}{l}{MSE} \\
 \cline{2-10}
 Scenario &  $\beta_0$ & $\beta_1$ & $\beta_2$ & $\gamma_0$ & $\gamma_1$ & $\gamma_2$ & $\nu$& $\rho_1$ & $\rho_2$\\
 \hline
1&	0.0028 	&	0.0020 	&	0.0023 	&	0.0039 	&	0.0033 	&	0.0037 	&	0.0004 	&	0.0008 	&	0.0023\\
2&	0.0028 	&	0.0023 	&	0.0026 	&	0.0041 	&	0.0034 	&	0.0035 	&	0.0003 	&	0.0008 	&	0.0038\\
3&	0.0025 	&	0.0023 	&	0.0024 	&	0.0037 	&	0.0032 	&	0.0035 	&	0.0004 	&	0.0007 	&	0.0027\\
4&	0.0025 	&	0.0022 	&	0.0023 	&	0.0035 	&	0.0030 	&	0.0037 	&	0.0004 	&	0.0008 	&	0.0035\\
5&	0.0022 	&	0.0020 	&	0.0024 	&	0.0038 	&	0.0030 	&	0.0039 	&	0.0003 	&	0.0008 	&	0.0024\\
6&	0.0028 	&	0.0022 	&	0.0025 	&	0.0035 	&	0.0029 	&	0.0037 	&	0.0004 	&	0.0009 	&	0.0034\\
7&	0.0022 	&	0.0016 	&	0.0022 	&	0.0033 	&	0.0022 	&	0.0031 	&	0.0003 	&	0.0005 	&	0.0014\\
8&	0.0025 	&	0.0014 	&	0.0022 	&	0.0031 	&	0.0018 	&	0.0032 	&	0.0002 	&	0.0005 	&	0.0015\\
9&	0.0022 	&	0.0015 	&	0.0022 	&	0.0029 	&	0.0020 	&	0.0032 	&	0.0003 	&	0.0005 	&	0.0013\\
10&	0.0023 	&	0.0013 	&	0.0021 	&	0.0034 	&	0.0023 	&	0.0029 	&	0.0003 	&	0.0006 	&	0.0016\\
11&	0.0023 	&	0.0015 	&	0.0023 	&	0.0038 	&	0.0021 	&	0.0033 	&	0.0002 	&	0.0004 	&	0.0013\\
12&	0.0025 	&	0.0014 	&	0.0022 	&	0.0036 	&	0.0021 	&	0.0030 	&	0.0003 	&	0.0004 	&	0.0016\\
\hline
\end{tabular}
\label{sim_mae}
\end{table}

\section{Real Data Analysis}\label{sec.8}

To see the usefulness of the proposed model for examining the dependence structure between (a) frequency and severity and (b) severities, we analyze a real automobile insurance dataset.

\subsection{Data}

We use the automobile insurance data provided by the Massachusetts Executive Office of Energy and Environmental Affairs, which were used by \citet{Ferreira}. The data contain the history of automobile insurance claims in 2006 in the state of Massachusetts.
With $15$ variables, the data consist of information on $3,991,012$ insured persons. The dataset also shows the $681,423$ claims of liability and personal injury protection coverage claim information with $10$ variables.
Each policyholder has information on the number of claims, individual claim amounts with the date of accidents, and covariates.
Among the observations, we randomly sample $1,500,000$ policyholders whose accidents occurred before 2008 and who have automobile insurance providing third party liability coverage for property damage and bodily injury.
The first $1,000,000$ observations are used as the training data to develop the model and the rest is reserved as the hold-out sample for validation purposes. 
We use the following two covariates, CLASS and TERRITORY, for the risk classification.
CLASS denotes five groups divided by policyholder characteristics (A: adults, B: business, I: $<3$ years of experience, M: $3\sim6$ years of experience, S: senior citizens).
TERRITORY denotes six territory groups divided by the driving characteristics (1: least risky to 6: most risky territory).
Table \ref{rev.tab.1} shows the mean frequency and mean severity per claim of the data categorized by CLASS and TERRITORY.
Full details of the covariates used can be found in the online supplement and in \citet{Ferreira}.

\begin{table}[h!]
\centering
\caption{Table of mean frequency and mean severity per claim (percentage of observations in brackets)}\label{rev.tab.1}
\begin{tabular}{|c|c|c|c|c|c|c|}
  \hline
   \multicolumn{2}{|c|}{}  & \multicolumn{5}{c|}{CLASS} \\
   \cline{3-7}
   \multicolumn{2}{|c|}{}  & A & B & I & M & S \\
\hline
    & \multirow{2}{*}{1} &   0.046 / 3241 &  0.046 / 1513 &  0.078 / 4186 & 0.058 / 4547  & 0.048 / 3426 \\
    &  &  (13.69\%) &  (0.31\%) & (1.17\%) & (0.87\%) & (2.89\%) \\
   \cline{2-7}
    & \multirow{2}{*}{2} &  0.046 / 3833 &  0.046 / 5516 & 0.070 / 3091 &  0.066 / 3501 & 0.049 / 3006 \\
    &  &  (13.82\%) &  (0.30\%) & (1.13\%) & (0.94\%) & (3.16\%) \\
   \cline{2-7}

\multirow{2}{*}{TERRITORY}    & \multirow{2}{*}{3} &  0.048 / 3933 &  0.033 / 2180 & 0.085 / 4724 & 0.055 / 3977 &  0.048 / 4196 \\
    &  &  (8.19\%) &  (0.15\%) & (0.58\%) & (0.55\%) & (1.67\%) \\
   \cline{2-7}

\multirow{2}{*}{}     & \multirow{2}{*}{4} &  0.051 / 3702 &  0.060 / 4431 & 0.097 / 3809 & 0.061 / 3601 & 0.060 / 4153 \\
    &  &  (14.88\%) &  (0.27\%) & (1.00\%) & (0.96\%) & (3.01\%) \\
   \cline{2-7}

    & \multirow{2}{*}{5} &  0.055 / 4042 &  0.083 / 2949 & 0.093 / 4052 & 0.084 / 4538 & 0.061 / 3812 \\
    &  &  (14.11\%) &  (0.21\%) & (0.84\%) & (0.94\%) & (2.72\%) \\
   \cline{2-7}
    & \multirow{2}{*}{6} &  0.062 / 4249 &  0.111 / 2979 & 0.107 / 4264 & 0.092 / 4968 & 0.070 / 4412 \\
    &  &  (8.87\%) &  (0.11\%) & (0.60\%) & (0.71\%) & (1.35\%) \\

  \hline
\end{tabular}

\end{table}

\subsection{Estimation results}

We apply Model \ref{mod.3} to the Massachusetts automobile data, where $N_{i}$ and $Y_{it}$ follow a Poisson distribution and a gamma distribution, respectively.
Table \ref{res} summarizes the estimation results for the model.
Based on the results, the class and territory groups are important factors for both the frequency and the severity parts.
The regression coefficients for the territory group show increasing patterns from the least risky to the most risky area in both the frequency and the severity parts.
The class group with less experience of driving (class=I) shows more accidents and the claim amount in that group tends to be higher than that in the other groups.
In the dependent collective risk model, the dependence between frequency and severity is measured by the parameter $\rho_1$. Its estimate is $-0.018$ with a 95$\%$ confidence interval of $(-0.031,-0.005)$, suggesting a significant negative correlation between the number of accidents and claim size.
Furthermore, this dependence seems weaker than the dependence between severities $\rho_2$.


\begin{table}[h!]
\caption{Estimation results}
\centering
\begin{tabular}{ l r r r r r r r  }
 \hline
Parameter& Est & Std. error & \multicolumn{2}{c}{95$\%$ CI} \\
 \hline
 \multicolumn{4}{l}{ {\bf Frequency part}} \\
\quad Intercept& 	-3.295&	0.012&	-3.319&	-3.270\\
\quad territory=2& 0.084&	0.016&	0.052&	0.116\\
\quad territory=3& 0.121&	0.018&	0.085&	0.157\\
\quad territory=4& 0.213&	0.016&	0.182&	0.243\\
\quad territory=5& 0.326&	0.015&	0.296&	0.356\\
\quad territory=6& 0.455&	0.016&	0.423&	0.487\\
\quad class=B& 	 0.306&	0.036&	0.235&	0.377\\
\quad class=I&  	 1.029&	0.016&	0.999&	1.060\\
\quad class=M& 	 0.497&	0.018&	0.462&	0.533\\
\quad class=S&  	-0.017&	0.013&	-0.043&	0.008\\
 \multicolumn{4}{l}{{\bf Severity part} } \\
\quad Intercept&   8.067&	0.014&	8.039&	8.095\\
\quad territory=2& 0.081&	0.019&	0.043&	0.118\\
\quad territory=3& 0.062&	0.022&	0.020&	0.104\\
\quad territory=4& 0.144&	0.018&	0.108&	0.179\\
\quad territory=5& 0.154&	0.018&	0.120&	0.189\\
\quad territory=6& 0.296&	0.019&	0.259&	0.333\\
\quad class=B&	 0.087&	0.042&	0.005&	0.170\\
\quad class=I& 	 0.115&	0.018&	0.080&	0.151\\
\quad class=M& 	 0.133&	0.021&	0.092&	0.175\\
\quad class=S& 	-0.105&	0.015&	-0.135&	-0.075\\
\quad $\nu$ &		 0.738&	0.004&	0.730&	0.746\\
 \multicolumn{4}{l}{{\bf Copula part}} \\
\quad $\rho_1$ &	-0.018&	0.006&	-0.031&	-0.005\\
\quad $\rho_2$ &	0.027&	0.001&	0.026&	0.029\\
\hline
\end{tabular}
\label{res}
\end{table}


To examine how well the proposed model and Tweedie's compound Poisson model
fit the training dataset,
we consider two quantities:
expected aggregate severity, $\E{S}$,
and
the value at risk of aggregate severity at the confidence level $\alpha= 0.995 $, $VaR_{0.995}(S)$,
by risk group defined by the CLASS and TERRITORY variables.
Figure \ref{compare_model2} and Table \ref{compare_model} report the results.
For comparison purposes, we also report the empirical values of $\E{S}$ and $VaR_{0.995}(S)$ (i.e., model-free estimates) for each risk group.
It is expected that good models produce $\E{S}$ and $VaR_{0.995}(S)$ close to the corresponding empirical values.
Figure \ref{compare_model2} shows that both models
provide similar estimates of $\E{S}$ and they are close to the empirical values.
Although the MSEs of the two estimates of $\E{S}$, at the bottom of Table \ref{compare_model}, do not show large differences,
the estimates of $VaR_{0.995}(S)$ of both models show substantial differences in MSEs.
Specifically, regarding $VaR_{0.995}(S)$, the risk group with class I shows high values within each territory group and it tends to be increasing as the territory becomes riskier.
This pattern is also found in Tweedie's model; however, it overestimates them for most of the risk groups, which makes its MSE larger than that of our proposed model.


\begin{figure}[h!]
\centering
\resizebox*{0.9\textwidth}{!}{\includegraphics{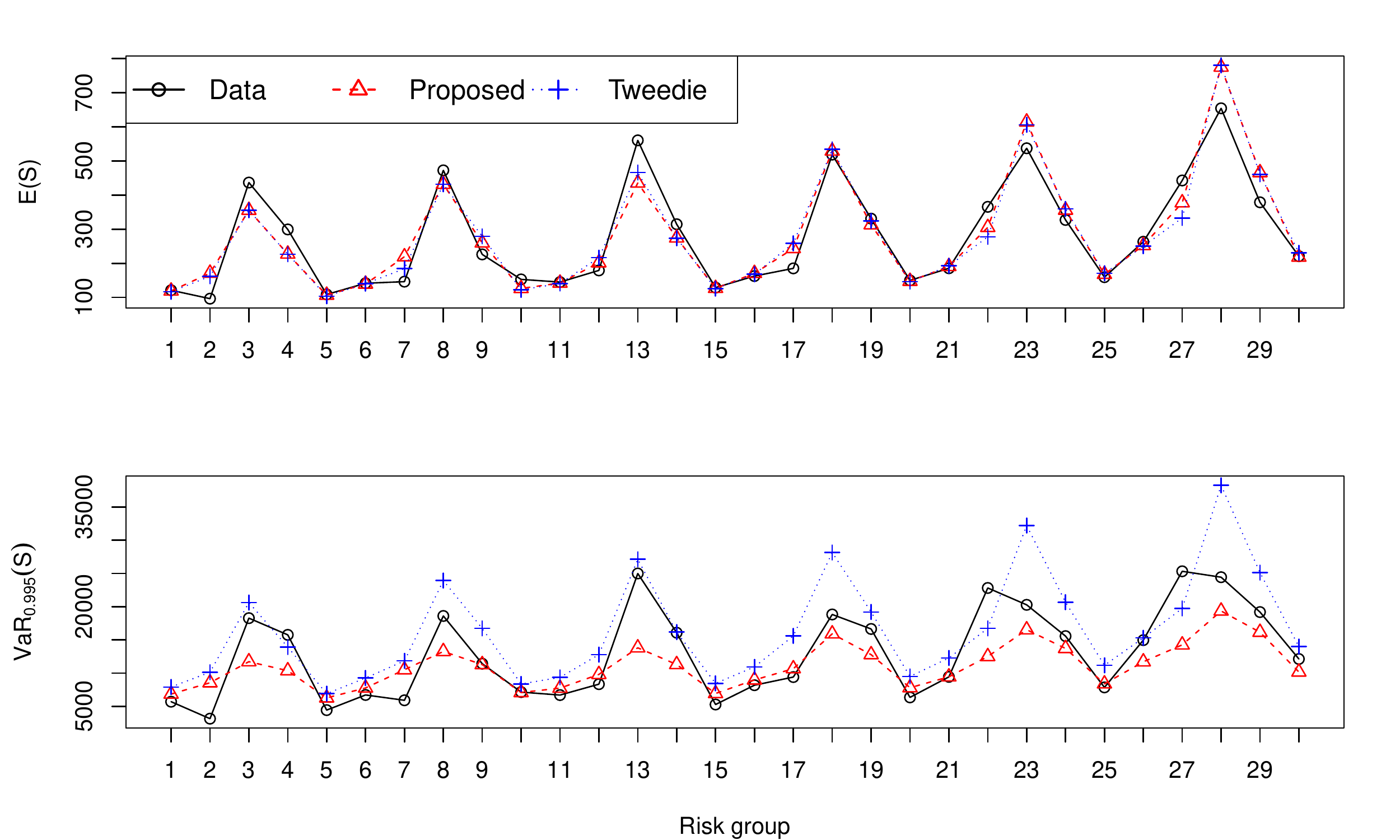} }
\caption{Comparison of $\E{S}$ and $VaR_{0.995}(S)$
by risk group for the training dataset from the Massachusetts automobile data: Empirical value, Proposed model, and Tweedie model}
\label{compare_model2}
\end{figure}


\begin{table}[p]
\caption{Comparison of $\E{S}$ and $VaR_{0.995}(S)$
by risk group for the training dataset from the Massachusetts automobile data: Empirical value, Proposed model, and Tweedie model}
\centering
\begin{tabular}{ l c c r r r r r r r r r r r r r r r r}
 \hline
 Risk & \multirow{2}{*}{$TERRITORY$}&  \multirow{2}{*}{$CLASS$}& \multicolumn{2}{l}{Data}&& \multicolumn{2}{l}{Proposed}&& \multicolumn{2}{l}{Tweedie}  \\ \cline{4-5} \cline{7-8} \cline{10-11}
Group & & & $E[S]$& $VaR_\alpha(S)$ && $E[S]$& $VaR_\alpha(S)$ && $E[S]$& $VaR_\alpha(S)$ \\
\hline
1	&	1	&	A	&	121 	&	5693 	&&	119 	&	6845 	&&	117 	&	7903 	\\
2	&	1	&	B	&	96 	&	3137 	&&	171 	&	8542 	&&	162 	&	10133 	\\
3	&	1	&	I	&	437 	&	18275 	&&	356 	&	11718 	&&	356 	&	20618 	\\
4	&	1	&	M	&	299 	&	15756 	&&	227 	&	10357 	&&	227 	&	13917 	\\
5	&	1	&	S	&	109 	&	4425 	&&	107 	&	6263 	&&	104 	&	6900 	\\
6	&	2	&	A	&	142 	&	6706 	&&	139 	&	7760 	&&	141 	&	9262 	\\
7	&	2	&	B	&	146 	&	5904 	&&	219 	&	10508 	&&	184 	&	11842 	\\
8	&	2	&	I	&	473 	&	18603 	&&	432 	&	13255 	&&	432 	&	23934 	\\
9	&	2	&	M	&	226 	&	11415 	&&	260 	&	11286 	&&	279 	&	16734 	\\
10	&	2	&	S	&	153 	&	7136 	&&	126 	&	7071 	&&	122 	&	8361 	\\
11	&	3	&	A	&	145 	&	6692 	&&	142 	&	7738 	&&	141 	&	9357 	\\
12	&	3	&	B	&	179 	&	8332 	&&	202 	&	9800 	&&	217 	&	12778 	\\
13	&	3	&	I	&	561 	&	25000 	&&	436 	&	13807 	&&	466 	&	27135 	\\
14	&	3	&	M	&	315 	&	16081 	&&	275 	&	11290 	&&	273 	&	16172 	\\
15	&	3	&	S	&	129 	&	5270 	&&	127 	&	6928 	&&	126 	&	8456 	\\
16	&	4	&	A	&	163 	&	8178 	&&	171 	&	8969 	&&	168 	&	10928 	\\
17	&	4	&	B	&	185 	&	9404 	&&	244 	&	10663 	&&	259 	&	15605 	\\
18	&	4	&	I	&	519 	&	18828 	&&	529 	&	15908 	&&	534 	&	28154 	\\
19	&	4	&	M	&	331 	&	16643 	&&	313 	&	12778 	&&	324 	&	19202 	\\
20	&	4	&	S	&	150 	&	6339 	&&	148 	&	7797 	&&	146 	&	9498 	\\
21	&	5	&	A	&	186 	&	9429 	&&	191 	&	9439 	&&	193 	&	12277 	\\
22	&	5	&	B	&	366 	&	22817 	&&	305 	&	12494 	&&	278 	&	16734 	\\
23	&	5	&	I	&	537 	&	20277 	&&	615 	&	16528 	&&	606 	&	32221 	\\
24	&	5	&	M	&	327 	&	15567 	&&	356 	&	13709 	&&	359 	&	20675 	\\
25	&	5	&	S	&	159 	&	7846 	&&	168 	&	8381 	&&	171 	&	11167 	\\
26	&	6	&	A	&	263 	&	14963 	&&	252 	&	11698 	&&	250 	&	15310 	\\
27	&	6	&	B	&	443 	&	25328 	&&	378 	&	14252 	&&	333 	&	19741 	\\
28	&	6	&	I	&	655 	&	24436 	&&	775 	&	19352 	&&	780 	&	38271 	\\
29	&	6	&	M	&	379 	&	19188 	&&	466 	&	16152 	&&	461 	&	25125 	\\
30	&	6	&	S	&	221 	&	12124 	&&	218 	&	10163 	&&	231 	&	13991 	\\
\hline\hline
\multicolumn{3}{l}{MSE} & \multicolumn{2}{c}{-}&& 	469  & 4192541 	&& 476 	& 9759489 \\
\hline
\end{tabular}
\label{compare_model}
\end{table}


\subsection{Validation results}

For validation purposes, we compare the two models in terms of the total loss prediction for the
500,000 policyholders in the hold-out sample.
Figure \ref{totalS} presents the predictive distributions of the various models, which are based on 5,000 Monte Carlo simulations under the estimation result from each model.
In the figure, the dotted vertical line indicates the actual amount of losses and dashed vertical line indicates the mean estimated total loss.
The predictive distribution from our proposed collective risk model has less variation, and its mean is closer to the actual total loss of the hold-out sample.
This result can be explained by the fact that there is a significant
negative correlation between frequency and individual severity, which is reflected appropriately in our model that allows such dependence.

\begin{figure}[h!]
\centering
\resizebox*{\textwidth}{!}{\includegraphics{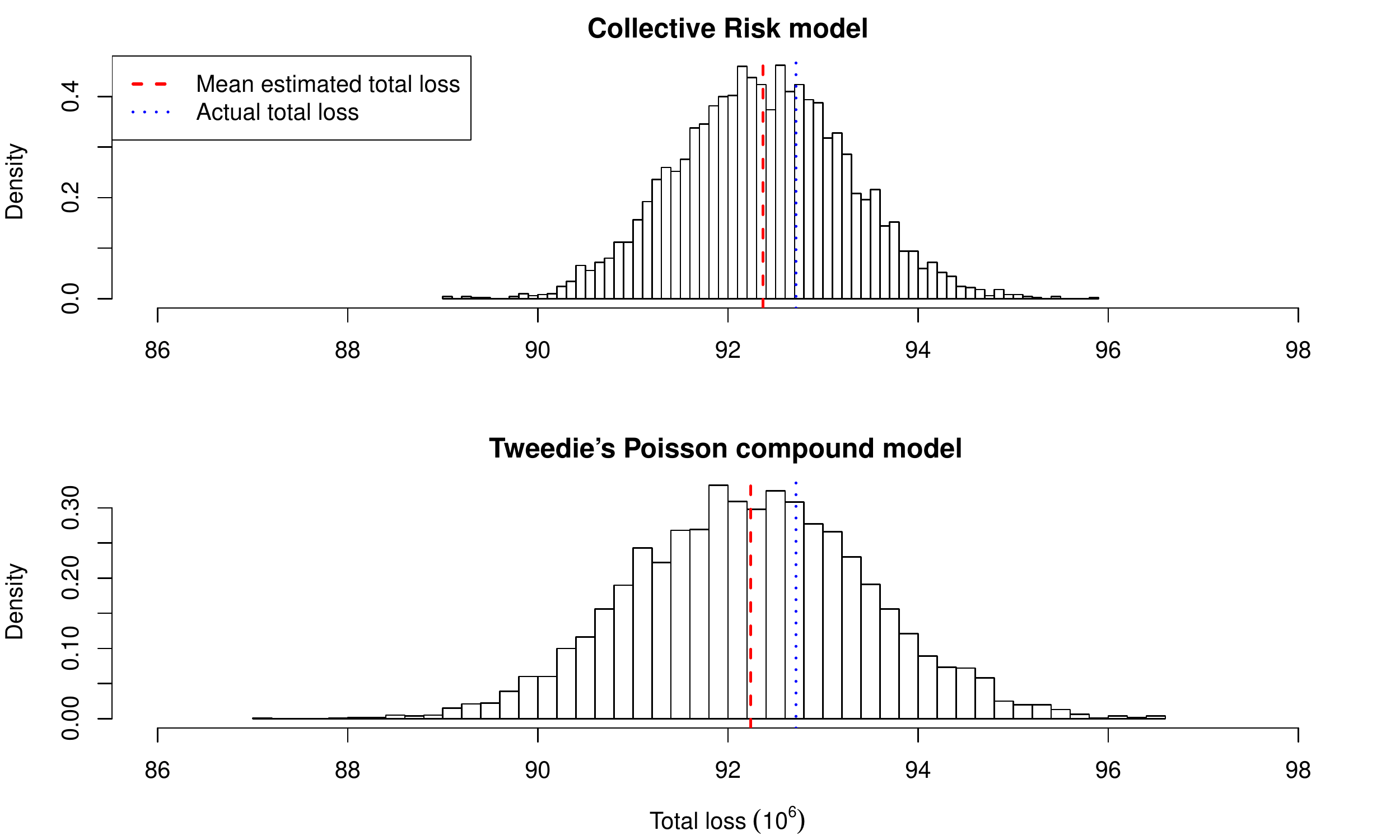} }
\caption{Predictive distribution of the total aggregated losses from the hold-out sample}
\label{totalS}
\end{figure}

\section{Conclusion}

We propose copula-based dependent collective risk models that allow the dependence between frequency and individual severities and
that among individual severities to be separate. We also provide the conditions for the two correlation matrices used to describe the dependence
to be positive definite.
In particular, we emphasize that using the Gaussian copula or t-copula has computational advantages because they allow an analytic form for
the conditional distribution of frequency given the individual severities.

Various extensions of our proposed models are possible as future research topics. First, it would be interesting to find appropriate general copula classes
that could be used for dependent collective risk models. Although we pursue some special copulas based on computational convenience,
if they cannot satisfactorily explain the given data, other complex copula functions may be necessary.
Consequently, the copula choice problem (i.e., model selection problem) becomes an important issue here. Second, it would
be interesting to model repeated measurements of frequency and individual severities over time.
Our copula model should thus be extended to take into account the fact that the measurements from different time points are correlated.
A promising approach would be to use a vine copula to combine several dependent collective risk models.

\newpage

\section*{Acknowledgements}
Woojoo Lee was supported by a Basic Science Research Program through the National Research Foundation of Korea (NRF) funded by the Ministry of Education (2016R1D1A1B03936100).
 Jae Youn Ahn was supported by a National Research Foundation of Korea (NRF) grant funded
by the Korean Government (NRF-2017R1D1A1B03032318).


\bibliographystyle{apalike}
\bibliography{Bib_Oh}

\appendix
\section{Proofs on Covariance Matrix}

\begin{proof}[Proof of Proposition \ref{prop.oh.1}]

  Proof of \eqref{eq.31} is a well known result from the elementary matrix algebra. For the proof of the second equation, we may use \eqref{eq.31} and Schur complement \citep{zhang2006schur}.

\end{proof}

\begin{proof}[Proof of Proposition \ref{prop.oh.5}]
Proof of \eqref{eq.32} is from the classical matrix algebra.
  For $k=0$, \eqref{eq.oh.110} is trivial. For $k>0$, with the following partitioned matrix representation
  \[
  \boldsymbol{\Sigma}_{\rho_1,\rho_2}^{[k,1]}=    \left(
    \begin{array}{cc}
      1 & \rho_1 \left({\bf 1}_k\right)^{\mathrm T}\\
      \rho_1 {\bf 1}_k & \boldsymbol{\Sigma}_{\rho_2}^{[k,1]}\\
    \end{array}
    \right)
  \]
  and Schur complement \citep{zhang2006schur},
  we have
  \begin{equation}\label{eq.oh.111}
  \left( \boldsymbol{\Sigma}_{\rho_1,\rho_2}^{[k,1]} \right)^{-1}=
  \left(
    \begin{array}{cc}
      1+ \left( \rho_1\boldsymbol{I}_k\right) \boldsymbol{E}^{-1}\left( \rho_1\boldsymbol{I}_k\right)^{\mathrm T}& - \left(\rho_1\boldsymbol{I}_k\right)^{\mathrm T} \boldsymbol{E}^{-1} \\
      -\boldsymbol{E}^{-1}\left( \rho_1\boldsymbol{I}_k\right) & \boldsymbol{E}^{-1} \\
    \end{array}
  \right)
  \end{equation}
  where
  \[
  \boldsymbol{E}:=\boldsymbol{\Sigma}_{\rho_2}^{[k,1]} - \left( \rho_1\boldsymbol{I}_k\right)\left( \rho_1\boldsymbol{I}_k\right)^{\mathrm T}.
  \]
  Furthermore, using \eqref{eq.32}, we have
\begin{equation}\label{eq.oh.112}
  \boldsymbol{E}^{-1}=\frac{1}{1-\rho_2}\left[\boldsymbol{I}_k -\frac{\rho_2-\rho_1^2}{1+(k-1)\rho_2-k\rho_1^2} \boldsymbol{J}_{k\times k}\right]
\end{equation}
  where the non-singularity of $\boldsymbol{E}^{-1}$ is guaranteed by ${\rm det}\left(\boldsymbol{\Sigma}_{\rho_1,\rho_2}^{[k,1]} \right)\neq 0$.
  Hence, simple algebraic manipulations using \eqref{eq.oh.111} and \eqref{eq.oh.112} conclude the proof.

\end{proof}

\begin{proof}[Proof of Theorem \ref{thm.1}]
  The first part is the classical result in matrix algebra. Now, we move to the second part. Since the proof is trivial for $k=0$ or $1$, we only consider a positive integer $k>1$.
  By Schur complement \citep{Zhang}, we have that $\boldsymbol{\Sigma}_{\rho_1,\rho_2}^{[k,1]}$ is positive definite if and only if
  \begin{equation}\label{pf.1}
  1-\left(\rho_1\boldsymbol{1}_k\right)^{\mathrm T}  \left(\boldsymbol{\Sigma}_{\rho_2}^{[k,1]}\right)^{-1} \left(\rho_1\boldsymbol{1}_k\right)
  \end{equation}
  and
  \begin{equation}\label{pf.10}
  \boldsymbol{\Sigma}_{\rho_2}^{[k,1]}
  \end{equation}
  are positive definite.

   From the following calculation
  \[
  -\left(1-\rho_2\right) \left( 1+(k-1)\rho_2\right) \left(1-\left(\rho_1\boldsymbol{1}_k\right)^{\mathrm T}  \left(\boldsymbol{\Sigma}_{\rho_2}^{[k,1]}\right)^{-1} \left(\rho_1\boldsymbol{1}_k\right)
  \right)
  = (1-\rho_2)\left(k\rho_1^2 -1 -(k-1)\rho_2\right),
  \]
  we have \eqref{pf.1} is positive definite if and only if \eqref{eq.10} holds. Furthermore, since \eqref{pf.10} is positive definite if and only if \eqref{eq.52} holds, we have that $\boldsymbol{\Sigma}_{\rho_1,\rho_2}^{[k,1]}$ is positive definite if and only if \eqref{eq.10} and \eqref{eq.52} holds. Now, we conclude the proof by observing the intersection of \eqref{eq.10} and \eqref{eq.52} is \eqref{eq.10}.

  \end{proof}

\begin{proof}[Proof of Proposition \ref{prop.oh.100}]
  The first equation is well known in the classical matrix algebra, and the second equation is from the first equation and Schur complement \citep{zhang2006schur}.
\end{proof}

\begin{proof}[Proof of Theorem \ref{thm.100}]
  The first part is the classical result in matrix algebra.
   Now, we move to the second part. Since the proof is trivial for $k=0$ or $1$, we only consider a positive integer $k>1$.

     By Schur complement \citep{Zhang}, we have that $\boldsymbol{\Sigma}_{\rho_1,\rho_2}^{[k,2]}$ is positive definite if and only if
  \begin{equation}\label{pf.17}
  1-\left(\rho_1\boldsymbol{1}_k\right)^{\mathrm T}  \left(\boldsymbol{\Sigma}_{\rho_2}^{[k,2]}\right)^{-1} \left(\rho_1\boldsymbol{1}_k\right)
  \end{equation}
  and
  \begin{equation}\label{pf.10a}
  \boldsymbol{\Sigma}_{\rho_2}^{[k,2]}
  \end{equation}
  are positive definite, where \eqref{pf.10a} is positive definite from the first part.

  From the following calculation
  \begin{equation*}
  -\left(1-\rho_2^2\right)\left( 1-\left(\rho_1\boldsymbol{1}_k\right)^{\mathrm T}  \left(\boldsymbol{\Sigma}_{\rho_2}^{[k,2]}\right)^{-1} \left(\rho_1\boldsymbol{1}_k\right) \right)
  = \rho_2^2 \left( (k-2)\rho_1^2 + 1 \right) - 2(k-1)\rho_1^2\rho_2+(k\rho_1^2-1),
  \end{equation*}
we have that \eqref{pf.17} is positive definite if and only if
  \begin{equation}\label{eq..2}
  \rho_2^2 \left( (k-2)\rho_1^2 + 1 \right) - 2(k-1)\rho_1^2\rho_2+(k\rho_1^2-1)<0.
  \end{equation}
 Since \eqref{eq..2} is evident for $k=2$, it is enough to prove \eqref{eq..2} for a positive integer $k>2$. From the following observation
 \[
 -1 <\frac{(k-2)\rho_1^2+\rho_1^2}{(k-2)\rho_1^2 + 1} <1
 \]
 and the fact that the left side of inequality in \eqref{eq..2} is a quadratic equation for $k>2$,
  we have \eqref{eq..2} if and only if
  \begin{equation}\label{eqf.27}
 g(\rho_1, \rho_2, k)<0, \quad\hbox{for}\quad \rho_2=-1, \quad 1, \quad \frac{(k-2)\rho_1^2+\rho_1^2}{(k-2)\rho_1^2 + 1}
  \end{equation}
 where
 \[
 g(\rho_1, \rho_2, k):=  \rho_2^2 \left( (k-2)\rho_1^2 + 1 \right) - 2(k-1)\rho_1^2\rho_2+(k\rho_1^2-1).
 \]
 Based on this, \eqref{eqf.27} always holds.

\end{proof}

\section{Auxilary Results on the Dependent Collective Risk Model}

For the proof of the Lemma 1, the following classical results on the conditional distribution in multivariate normal distribution are applied to the special covariance structure
$\boldsymbol{\Sigma}_{\rho_1, \rho_2}^{[k,1]}$.

\begin{lemma}\label{cor.2}
Let $k_0\in\mathbb{N}_+$, and consider $\rho_1, \rho_2\in(-1,1)$ satisfying
\eqref{con.oh.21} and \eqref{c.o.1} for $z=1$ and $z=2$, respectively.
If we consider
  $
  \left(Y_0, Y_1, \cdots, Y_k \right){^\mathrm{T}}\sim {\rm MVN}\left(\boldsymbol{0}_{k+1}, \boldsymbol{\Sigma}_{\rho_1, \rho_2}^{[k, z]} \right)$
  with positive integer $k\le k_0$ and for $z\in\{1,2\}$, the conditional distribution of $Y_0$ for given $(Y_1, \cdots, Y_{k})=(y_1, \cdots, y_{k})$ is
      \[
      Y_0\big\vert  (Y_1, \cdots, Y_{k})=(y_1, \cdots, y_{k}) \sim {\rm N}(\mu_1, \sigma_1^2)
      \]
      where
      \begin{equation*}
       \mu_1=\left( \rho_1\boldsymbol{1}_n^{\mathrm T}\right) \left(\boldsymbol{\Sigma}_{\rho_2}^{[k, z]}\right)^{-1} (y_1, \cdots, y_{k})^\mathrm{T} \quad\hbox{and}\quad
         \sigma_1^2=1- \left( \rho_1\boldsymbol{1}_k^{\mathrm T}\right) \left(\boldsymbol{\Sigma}_{\rho_2}^{[k, z]}\right)^{-1} \left( \rho_1\boldsymbol{1}_k\right).
      \end{equation*}
\end{lemma}

For the proof, see \citet{Johnson}.

\begin{lemma}\label{lem.oh.1}

Consider Conditional Model \ref{mod.1} for for $z\in\{1,2\}$. The conditional distribution of $N$ is given as
\[
\P{N\le n \big\vert \boldsymbol{y}^{[k]}}=
\Phi\left(\frac{\Phi^{-1}(F_1(n))-\mu_{[k,1]}}{\sigma_{[k,1]}} \right)
\]
for a positive integer $k$ and $\boldsymbol{y}^{[k]}\in \Real_+^k$,
where $\mu_{[k, z]}$ and $\sigma_{[k, z]}$ are defined in \eqref{eq.18} and \eqref{eq.18.a}.
\end{lemma}
\begin{proof}
  For convenience, define
  \[
  \mathcal{I}_{N}:=\{n\in \mathcal{N}_0\big\vert \P{N=n}>0\}.
  \]

  For the calculation of $$\P{N\le n\big\vert y_1, \cdots, y_k}$$
  let $N^*\sim F_1^*$ be continuous cumulative distribution satisfying
  \[
  F_1^*(-1)=0 \quad\hbox{and}\quad F_1^*(n)=F_1(n)
  \]
  for positive integer $n\in \mathcal{I}_{N}$
  where $F_1^*(\cdot)$ being a strictly increasing function on $(-1, N_{\sup})$, where $N_{\sup}$ is defined as the essential supremum of $N$. Existence of such $F_1^*$ is guaranteed by linear interpolation on the fixed points
  \[
  \{(-1,0)\} \cup \{ (n, F_1(n)) \big\vert n\in  \mathcal{I}_{N} \}.
  \]

  Now, consider the joint distribution function of $N^*$ and $\boldsymbol{Y}^{[k]}$
  \[
  \P{(N^*, Y_1, \cdots, Y_k)\le (n, y_1, \cdots, y_k)}:=C_{\boldsymbol{\Sigma}_{\rho_1, \rho_2}^{[k,1]}}(F_1^*(n), F_2(y_1), \cdots, F_2(y_k))
  \]
  for any integer $n$ and $(y_1, \cdots,y_k)\in\Real^{k}$.
Then, from Lemma \ref{cor.2}, we have
\[
\left( \Phi^{-1}(F_1^*(N^*)), \Phi^{-1}(F_2(Y_1)), \cdots, \Phi^{-1}(F_2(Y_k)) \right)\sim {\rm MVN}\left( \boldsymbol{0}_{z+1}, \boldsymbol{\Sigma}_{\rho_1, \rho_2}^{(z)}\right)
\]
and
\[
\Phi^{-1}(F_1^*(N^*))\big\vert Y_1=y_1, \cdots, Y_k=y_k \sim {\rm N}(\mu_{[k,1]}, (\sigma_{[k,1]})^2)
\]
where $\mu_{[k,1]}$ and $\sigma_{[k,1]}$ are defined in \eqref{eq.18} and \eqref{eq.18.a}.

Hence, we have
\begin{equation}\label{eq.oh.100}
\begin{aligned}
  \P{N\le n \big\vert \boldsymbol{y}^{[k]}}
   &=\frac{\frac{\partial^z}{\partial y_1\cdots, \partial y_k}C_{\boldsymbol{\Sigma}_{\rho_1, \rho_2}^{[k_2]}} \left( F_1(n), F_2(y_1), \cdots, F_2(y_k) \right)}{f_{\boldsymbol{Y}}(\boldsymbol{y}^{[k]}) }\\
   &=\frac{\frac{\partial^z}{\partial y_1\cdots, \partial y_k}C_{\boldsymbol{\Sigma}_{\rho_1, \rho_2}^{[k_2]}} \left( F_1^*(n), F_2(y_1), \cdots, F_2(y_k) \right)}{f_{\boldsymbol{Y}}(\boldsymbol{y}^{[k]}) }\\
   &=\P{N^*\le n \big\vert \boldsymbol{y}^{[k]}}\\
   &=\P{\frac{\Phi^{-1}(F_1^*(N^*)) - \mu_{[k,1]}}{\sigma_{[k,1]}} \le \frac{\Phi^{-1}(F_1^*(n))-\mu_{[k,1]}}{\sigma_{[k,1]}} { \bigg\vert \boldsymbol{y}^{[k]} } }\\
   &=\Phi\left(\frac{\Phi^{-1}(F_1^*(n))-\mu_{[k,1]}}{\sigma_{[k,1]}} \right)\\
   &=\Phi\left(\frac{\Phi^{-1}(F_1(n))-\mu_{[k,1]}}{\sigma_{[k,1]}} \right).\\
\end{aligned}
\end{equation}
for an integer $n\in \mathcal{I}_{N}\cup \{-1\}$.
 Finally, the following observation with \eqref{eq.oh.100}
\[
\P{N\le n \big\vert \boldsymbol{y}^{[k]}} = \P{N\le n-1 \big\vert \boldsymbol{y}^{[k]}}\quad\hbox{for positive integer $n\notin \mathcal{I}_{N}\cup \{-1\}$ }
\]
concludes the proof.
\end{proof}

\begin{proof}[Proof of Lemma \ref{prop.2}]
We start from the proof of part i.
First, for $z=0$, we have
 \[
 h_{N, \boldsymbol{Y}^{[k]}}(n, \boldsymbol{y}^{[k]})=f_1(0)
 \]
 by definition in Conditional Model \ref{mod.1}.

Now, let $z$ be any positive integer. Then, the joint density function of the discrete margin $N$ and the continuous margins $Y_1, \cdots, Y_k$ is given by
\[
h_{N, \boldsymbol{Y}^{[k]}}(n, \boldsymbol{y}^{[k]}):= {\rm P}_k\left(n, \boldsymbol{y}^{[k]}\right)-  {\rm P}_k\left(n-1, \boldsymbol{y}^{[k]} \right)
\]
for integer $n\in\mathbb{Z}$ and $\boldsymbol{y}^{[k]}\in\Real^k$, where
 \[
 {\rm P}_k\left(n, y_1, \cdots, y_k \right):=\frac{\partial^z}{\partial y_1\cdots, \partial y_k}\P{N\le n, Y_1\le y_1, \cdots, Y_k\le y_k}.
 \]
Here,  ${\rm P}_k\left(n, y_1, \cdots, y_k \right)$ can be written as
\begin{equation}\label{eq.oh.3}
 {\rm P}_k\left(n, \boldsymbol{y}^{[k]} \right)=\P{N\le n\big\vert y_1, \cdots, y_k}f_{\boldsymbol{Y}^{[k]}}(y_1, \cdots, y_k)
\end{equation}
where $f_{\boldsymbol{Y}^{[k]}}$ is the density function of $\boldsymbol{Y}^{[k]}$.
Using Lemma \ref{lem.oh.1}, we have
\[
\begin{aligned}
  h_{N, \boldsymbol{Y}^{[k]}}(n, \boldsymbol{y}^{[k]})&=\P{N\le n \big\vert \boldsymbol{y}^{[k]}}f_{\boldsymbol{Y}}(\boldsymbol{y}^{[k]})-\P{N\le n-1 \big\vert \boldsymbol{y}^{[k]}}f_{\boldsymbol{Y}}(\boldsymbol{y}^{[k]})\\
  &=f_{\boldsymbol{Y}^{[k]}}(\boldsymbol{y}^{[k]})
      \left(
      \Phi\left( \frac{\Phi^{-1}(F_1(n)) -\mu_{[k,1]}}{\sigma_{[k,1]}} \right) - \Phi\left( \frac{\Phi^{-1}(F_1(n-1)) -\mu_{[k,1]}}{\sigma_{[k,1]}} \right)
      \right)
\end{aligned}
\]
for an integer $n\in\mathbb{Z}$ and $\boldsymbol{y}^{[k]}\in\Real^k$.

Note that since a continuous random vector $\boldsymbol{Y}^{[k]}$ follows
 $$\boldsymbol{Y}^{[k]}\sim C_{\boldsymbol{\Sigma}_{\rho_2} ^{(z)}}(F_2, \cdots, F_2) $$
 for any positive integer $z$, then the density function of $\boldsymbol{y}^{[k]}$ is given as in \eqref{eq.oh.6}. The proof of Part i is complete. For brevity, we omit the proof of Part ii because it is trivial from Part i.
\end{proof}

\begin{lemma}\label{cal.lem}
      Consider the frequency and severities
\[
(N,\boldsymbol{Y}_N):=\begin{cases}
  (N, Y_1, \cdots, Y_N), & N\ge 1;\\
  0, & N=0;\\
\end{cases}
\]
in a dependent collective risk model in Model \ref{mod.2} with the assumptions employed in Conditional Model \ref{mod.1} or Conditional Model \ref{mod.10}.
Then, for a positive integer $n$ and $k$ satisfying $k\le n$, we have
\begin{equation}\label{eq.29}
  \E{Y_k\big\vert N=n} \P{N=n\big\vert N>0} = \int_0^\infty y\,
  h_{N^+, \boldsymbol{Y}^{[1]}}(n, y) {\rm d} y
\end{equation}
and
\begin{equation}\label{eq.30}
  \E{S\big\vert N=n} \P{N=n\big\vert N>0} = n
  \int_0^\infty y
  h_{N^+, \boldsymbol{Y}^{[1]}}(n, y) {\rm d} y
\end{equation}
where  $h_{N^+, \boldsymbol{Y}^{[1]}}$ is defined in \eqref{eq.14} or \eqref{eq.141}, depending on Conditional Model \ref{mod.1} or Conditional Model \ref{mod.10} assumptions as well as the type of the covariance matrix.

\end{lemma}
\begin{proof}
  For a positive integer $n$ and $k$ satisfying $k\le n$, we have
  \[
   \begin{aligned}
     \E{Y_k\big\vert N=n} \P{N=n\big\vert N>0}
     &= \E{Y_k\big\vert R=1, N^+ = n} \P{R=1, N^+=n\big\vert R=1}\\
     &= \E{Y_k\big\vert R=1, N^+ = n} \P{N^+=n\big\vert R=1}\\
     &= \E{Y_k\big\vert N^+ = n} \P{N^+=n}\\
     &= \int_0^\infty y\, h_{N^+, \boldsymbol{Y}^{[1]}}(n, y) {\rm d} y
   \end{aligned}
  \]
  where the third equality is from the independence assumption between $R$ and $(N^+, Y_1, \cdots, Y_k)$ for any positive integer $k$ satisfying $k<n_{\sup}+1$, and the last equality is from Lemma \ref{prop.2}. Finally, the equation \eqref{eq.29} derives
  \[
   \begin{aligned}
     \E{S\big\vert N=n} \P{N=n\big\vert N>0}
     &= \sum\limits_{j=1}^{n}\E{Y_j\big\vert N = n} \P{N=n\big\vert N>0}\\
     &= n \int_0^\infty y\, h_{N^+, \boldsymbol{Y}^{[1]}}(n, y) {\rm d} y.
   \end{aligned}
  \]
\end{proof}

\section{Proofs on Proposition 5}

\begin{proof}
  The proof of the first part is from the following equation
  \[
  \begin{aligned}
    \E{S} &= \E{\E{S\big\vert N}}\\
    &=\sum\limits_{n=0}^{\infty}\E{S\big\vert N=n}\P{N=n}\\
    &=\sum\limits_{n=1}^{\infty}\E{S\big\vert N=n}\P{N=n}\\
    &=\sum\limits_{n=1}^{\infty}\E{S\big\vert N=n} \P{N^+=n, R=1}\\
    &=p \sum\limits_{n=1}^{\infty}\E{S\big\vert N=n} \P{N=n\big\vert N>0}\\
    &=p \sum\limits_{n=1}^{\infty} n \int_0^\infty y\, h_{N^+, \boldsymbol{Y}^{[1]}}(n, y) {\rm d} y.
  \end{aligned}
  \]
  where the last equality is from \eqref{eq.30}.
   The proof of the second part follows from the following equation
  \[
  \begin{aligned}
  \E{NS} &= \E{\E{NS\big\vert N}}\\
  &= \E{N\E{S\big\vert N}}\\
  &= \sum\limits_{n=1}^{\infty} n \E{S\big\vert N=n} \P{N=n}\\
  &= p\sum\limits_{n=1}^{\infty} n \E{S\big\vert N=n} \P{N=n\big\vert N>0}\\
  &= p\sum\limits_{n=1}^{\infty} n^2 \int_0^\infty y
  h_{N^+, \boldsymbol{Y}^{[1]}}(n, y) {\rm d} y
  \end{aligned}
  \]
  where the last equality is from \eqref{eq.30}.

  Finally, the proof of the second part is from the following observations:
  \[
  \begin{aligned}
  \E{NM\big\vert N>0} &= \E{S\big\vert N>0}\\
  &= \sum\limits_{n=1}^{\infty} \E{S\big\vert N>0, N=n} \P{N=n\big\vert N>0}\\
  &= \sum\limits_{n=1}^{\infty} \E{S\big\vert N=n} \P{N=n\big\vert N>0}\\
  &= \sum\limits_{n=1}^{\infty} n \int_0^\infty y
  h_{N^+, \boldsymbol{Y}^{[1]}}(n, y) {\rm d} y
  \end{aligned}
  \]
  where the last equality is from \eqref{eq.30} and
  \[
  \begin{aligned}
  \E{M\big\vert N>0} &= \E{\frac{1}{N} \sum\limits_{j=1}^{N} Y_j\big\vert N>0}\\
  &= \sum\limits_{n=1}^{\infty} \frac{1}{n} \sum\limits_{j=1}^{n} \E{  Y_j\big\vert N>0, N=n} \P{N=n \big\vert N>0}\\
  &= \sum\limits_{n=1}^{\infty} \frac{1}{n} \sum\limits_{j=1}^{n} \E{  Y_j\big\vert  N=n} \P{N=n \big\vert N>0}\\
  &= \sum\limits_{n=1}^{\infty} \frac{1}{n} \sum\limits_{j=1}^{n} \int_0^\infty y\,
  h_{N^+, \boldsymbol{Y}^{[1]}}(n, y) {\rm d} y\\
  &=  \sum\limits_{n=1}^{\infty} \int_0^\infty y\,
  h_{N^+, \boldsymbol{Y}^{[1]}}(n, y) {\rm d} y
  \end{aligned}
  \]
  where the last equality is from \eqref{eq.29}.

\end{proof}

\end{document}